\RequirePackage{fix-cm}
\documentclass[twocolumn,smallextended]{svjour3}       
\smartqed  
\usepackage{graphicx}
\usepackage{amsmath} 
\usepackage{amsfonts}
\usepackage{amssymb}

\usepackage{algorithm}
\usepackage{algpseudocode}
\usepackage{subcaption}
\usepackage{breqn}
\usepackage{hyperref}
\usepackage{float}
\usepackage{booktabs}
\usepackage{array}
\newcolumntype{L}[1]{>{\raggedright\let\newline\\\arraybackslash\hspace{0pt}}m{#1}}
\newcolumntype{C}[1]{>{\centering\let\newline\\\arraybackslash\hspace{0pt}}m{#1}}
\newcolumntype{R}[1]{>{\raggedleft\let\newline\\\arraybackslash\hspace{0pt}}m{#1}}

\begin{document}\sloppy
	
	\title{Tracking Mobile Intruders in an Art Gallery: Guard Deployment Strategies, Fundamental Limitations, and Performance Guarantees
	}
	
	
	\author{Guillermo J. Laguna \and Sourabh Bhattacharya   
	}
	
	
	\institute{Guillermo J. Laguna \at
		Department of Mechanical Engineering \\
		Iowa State University\\
		Ames, IA USA
		Tel.: +521-462-1905156\\
		\email{gjlaguna@iastate.edu}           
		\and
		Sourabh Bhattacharya \at
		Department of Mechanical Engineering \\
		Iowa State University\\
		Ames, IA USA
		\email{sbhattac@iastate.edu}
	}

	\date{Received: date / Accepted: date}

	\maketitle
	
	\begin{abstract}
		This paper addresses the problem of tracking mobile intruders in a polygonal environment. We assume that a team of diagonal guards is deployed inside the polygon to provide mobile coverage. First, we formulate the problem of tracking a mobile intruder inside a polygonal environment as a multi-robot task allocation (MRTA) problem. Leveraging on guard deployment strategies in art gallery problems for mobile coverage, we show that the problem of finding the minimum speed of guards to persistently track a single mobile intruder is NP-hard. Next, for a given maximum speed of the intruder and the guards, we propose a technique to partition a polygon, and compute a feasible allocation of guards to the partitions. We prove the correctness of the proposed algorithm, and show its completeness for a specific class of inputs. We classify the guards based on the structural properties of the partitions allocated to them. Based on the classification, we propose motion strategy for the guards to track the mobile intruder when it is located in the partition allocated to the guard. Finally, we extend the proposed technique to address guard deployment and allocation strategies for non-simple polygons and multiple intruders.
	\end{abstract}
\section{Introduction}
\label{sec:intro}
Security is an important concern in infrastructure systems. For decades, autonomous mobile robots have been utilized as surveillance \cite{Mittal:2016}\cite{Witwicki:2017} and crime-fighting agents for barrier assessments \cite{Theodoridis:2012}\cite{Ciccimaro:1999}, intruder detection \cite{Tuna:2012}\cite{Bazydlo:2017}, building virtual terrains or maps \cite{Oriolo:1998}\cite{Montemerlo:2006}, neutralizing explosives \cite{Majerus:2010}\cite{Nicoud:1995}, and recognizing abnormal human behaviors \cite{Rezazadegan:2017}\cite{Popoola:2012}. Such robots have been designed with the ability to counter threats, limit risks to personnel, and reduce manpower requirements in hazardous environments \cite{Lee:2003}\cite{Yamamoto:1992}\cite{Theodoridis:2012}. Deployment of these surveillance platforms in teams has given rise to challenging problems in collaborative sensing and decision-making. Motivated from recent surge in interest in autonomous surveillance, we address an asset protection problem in which a team of mobile sensors collaborate to track suspicious mobile entities to secure an environment. 

Although advanced electronic and biometric techniques can be used to secure facilities, vision-based monitoring is widely used for persistent surveillance. The idea is to visually cover the environment in order to obtain sufficient information so that appropriate measures can be taken to secure the area in case of any suspicious/malicious activity. The general formulation of a tracking problem consists of a team of autonomous sensing platforms, called {\it observers}, that visually track mobile entities, called {\it targets}. In this work, we consider the aforementioned formulation in the presence of features in the environment that can occlude the targets from the observers, for example, presence of reflex vertices on the boundary of the environment and obstacles. Precocious planning and coordination between observers can prevent the intruders from escaping the visual footprint of the observers around such occlusions.

Deploying a network of autonomous sensing platforms has been an active area of research in multi-robot systems (MRS). Multi-robot systems have emerged as an important area of research in robotics due to  their potential applications in several areas, for example, autonomous sensor networks \cite{Ghosh:2017}, building surveillance \cite{Dipaola:2010}, transportation of large objects \cite{Sakuyama:2012}, air and underwater pollution monitoring \cite{Dunbabin:2012}, forest fire detection \cite{Afzaal:2016}, transportation systems \cite{Matsuhira:2010}, or search and rescue after large-scale disasters \cite{Ko:2009}. Even problems that can be handled by a single multi-skilled robot may benefit from the alternative usage of a robot team, since robustness and reliability can often be increased by combining several robots which are individually less robust and reliable. In case of target tracking, multiple points of view from multiple robots add extra information on the target resulting in a better estimate of its position. However, the uncertainty in the future actions of the target, and the tight coupling between sensing, coordination and control within the team of mobile observers gives rise to challenging problems in multi-robot motion planning.

A simple solution to the tracking problem in bounded environments is to cover the environment with sufficient number of observers. This leads to the {\it art gallery problem}, a well-studied topic in computational geometry.  In the classical art gallery problem, the goal is to determine the number of guards needed to visually cover a bounded polygon \cite{Fisk:1978}. A guard can only see the portion of the polygon unobstructed by the boundary of the polygon or by internal obstacles. Over the years several variants of the problem have been studied based on the shape of the polygon (orthogonal \cite{Kahn:1983,Durocher:2017,Durocher:2014,Durocher:2015}, monotone \cite{Deberg:2017}, etc), the type of guards employed (static \cite{O'Rourke:1987,Kahn:1983,Hoffmann:1990,Czyzowicz:1994,Urrutia:2000}, mobile \cite{O'Rourke:1983,Toussaint:1982,Colley:1995,Kay:1970}), and the notion of visibility ($k$-transmitters \cite{Mahdavi:2014}, multi-guarding \cite{Fazli:2010,Fazli:2010a} etc). Stationary guards can either be point guards (placed anywhere inside the polygon) or vertex guards (restrict them only to the vertices of the polygon). It has been shown that the problem of computing the minimum number of stationary guards required to cover a simply connected polygon is NP-hard \cite{Lee:1986}.  In \cite{O'Rourke:1987}, it is shown that $\lfloor n/3 \rfloor$ static guards with omni-directional field-of-view is sufficient, and sometimes necessary to cover the entire polygon, where $n$ is the number of vertices of the polygon representing the environment. Efforts to obtain a bound on the minimum number of guards required to cover a polygon for special cases include approximation \cite{Ghosh:2010} and heuristic techniques \cite{Amit:2010}. 

The notion of mobile guards was first introduced by Toussaint \cite{Toussaint:1982} where each guard can travel back and forth along a segment inside the polygon, and every point in the polygon must be seen by at least one guard at some point of time along its path. Since the guards do not visually cover the entire polygon at all times, this notion of coverage is called {\it mobile coverage}. In \cite{O'Rourke:1983} it is proved that $\lfloor n/4 \rfloor$ diagonal guards are sufficient to provide mobile coverage. Therefore, if the guards have sufficient speed, at most $\lfloor n/4 \rfloor$ of them are required to track a mobile intruder in a polygonal environment. This naturally leads to the following question: What is the minimum speed required for the diagonal guards to track a mobile intruder in the environment with known maximum speed? The answer depends on the coordination and control strategy used by the guards.

The need for coordination arises in multi-robot systems (MRS) when individual agents work together to complete a common task \cite{Barnes:1991}. Depending on the communication scheme adopted by the MRS, coordination can be explicit or implicit. As the name suggests, explicit communication requires the presence of an on-board communication module on each robot of the MRS. On the other hand, implicit communication is generally contextual, based on the sensor data available to a robot regarding its teammates. Due to significant improvements in mobile communication hardware, the coordination scheme prevalent in current MRSs is predominantly explicit in nature.  In this work, we propose an explicit coordination scheme for the team of observers. Our scheme explicitly relies on the location of the observers relative to each other thereby building a connection between the geometric aspects of the team formation, and communication topology required for persistent tracking.

In this work, we formulate the tracking problem as a multi-robot task allocation problem (MTAP) \cite{Gerkey:2004a}. The allocation tries to balance the workload among the observers thereby minimizing the speed required to ensure tracking. The main contributions of this work are as follows:

\begin{enumerate}
	\item We investigate a variant of the multi-robot target tracking problem in which observers are constrained to move along diagonals of the polygonal environment. We leverage guard deployment strategies proposed for mobile coverage \cite{O'Rourke:1983} to design deployment and tracking strategies for multiple observers. To the best of our knowledge, our work is the first to build a connection between the well known art gallery problem which deals with coverage and multi-robot tracking. 
	\item  We propose an algorithm that jointly partitions a polygon, and allocates observers to track a target inside the partition. The algorithm is a solution to a resource allocation problem in which the observers are resources for tracking that need to be assigned to partitions of the polygon. The joint partitioning and allocation algorithm is based on the triangulation of the polygon, and therefore, can be extended to polygon with holes. Under the restriction that a single observer can be allocated to each partition, we show that the problem of finding the minimum observer speed required to track an intruder is NP-hard. 
	\item We present a taxonomy of the observers based on their task allocation, and the level of cooperation needed from other members of the team to perform the tracking task. We introduce the notion of {\it critical curves} to construct activation strategies for the mobile observers.  
	
	\item We derive the maximum number of targets that can be tracked using the deployment and partitioning algorithm proposed in this work.
\end{enumerate}


The paper is organized as follows. Section \ref{sec:rel} presents a brief description of the related work. Section \ref{sec:prob} presents the problem formulation and a deployment strategy for the guards. Section \ref{sec:alloc} formulates the tracking problem as a multi-robot task allocation problem. Section \ref{sec:gen} presents a partitioning technique for the polygon, and an allocation algorithm to assign the partitions to the guards to track a single intruder in the environment. Section \ref{sec:holes} addresses the case of non-simple polygons. Section \ref{sec:multi} extends the proposed deployment and allocation algorithm to address the case of multiple intruders in the environment. Section \ref{sec:conclusion} presents the conclusions and future work. 

\section{Related Work}
\label{sec:rel}
The multi-robot target-tracking problem was originally introduced by Parker and Emmons in \cite{Parker:1997}. The authors proposed the framework of Cooperative Multi-robot Observation of Multiple Moving Targets (CMOMMT) for a team of observers that simultaneously maximize the number of targets under observation, and the duration of observation of each target. It is shown that the CMOMMT problem is NP-hard. Consequently, several variants of CMOMMT have been studied to propose implementable solutions, for example, Approximate CMOMMT \cite{Parker:2002}, personality-CMOMMT \cite{Ding:2006}, behavioral-CMOMMT \cite{Kolling:2006,Kolling:2007}, formation-CMOMMT \cite{Werger:2000}, to name a few. Alternate formulations of the CMOMMT based on particle filtering \cite{Spletzer:2003} and mixed non-linear integer programming formulation \cite{Xu:2013a} have also been proposed in the past. These frameworks lead to a numerical approach for generating the trajectory of the observers.

In the past, several variants of the original tracking problem posed in \cite{Parker:1997} have been addressed. In \cite{Isler:2005}, the authors analyze the {\it focus of attention} problem \cite{Goossens:2016} for specific formations of the observers, and provide approximation algorithms for optimal target allocation. Dames et al. address the problem of detecting, localizing and tracking a team of targets with unknown and time-varying cardinality \cite{Dames:2018}. In \cite{Hausman:2015}, the authors address the cooperative control of a team of UAVs that try to solve the joint problem of self localization and tracking with on-board sensors. In \cite{Jung:2002}, a region-based coarse approach is proposed to simultaneously observe several mobile targets. The observers use a local method to maximize the number of observed targets. The proposed technique does not require any communication among the robots. In \cite{Markov:2007}, CMOMMT is addressed for cooperative targets. The targets are able to communicate which allows active cooperation through sharing data. Thus the average observation time is minimized through a force field approach. In \cite{Tang:2005}, a team of winged UAVs tries to minimize the average time elapsed between two consecutive observations of each member of a group of targets. An optimal control scheme is used to obtain the motion strategy of the observers.

In this work, we formulate the problem of tracking as a multi-robot task allocation problem. Multi-robot task allocation (MRTA) can be considered as an instance of the well-known optimal assignment problem. In \cite{Nam:2014}, algorithms to solve the matching problem for weighted bipartite multi-graphs are used to solve MTAP. In \cite{Gerkey:2003} a domain-independent taxonomy of multi-robot task allocation problems is presented. They also analyzed and compared some iterated assignment architectures: ALLIANCE \cite{Parker:1999}, BLE \cite{Werger:2000}, and M+ \cite{Botelho:1999} and some on-line assignment architectures: MURDOCH (auction-based MRTA) \cite{Gerkey:2001}, first-price auctions (market-based MRTA) \cite{Stentz:1999} and dynamic role assignment \cite{Chaimowicz:2002}, for MRTA, respectively. In \cite{Lerman:2006}, the authors present a mathematical modeling and analysis of the collective behavior of dynamic task allocation. In \cite{Yan:2011}, a lightweight and robust multi-robot task allocation approach based on trade rules in market economy is presented. Other strategies are based on vacancy chains \cite{Dahl:2009}, auction-based mechanisms \cite{Hanna:2005} and distributed market-based coordination \cite{Michael:2008}.

In this work, we assume that the observers are omni-directional cameras that slide along the diagonals of the polygon. 
In \cite{Durocher:2017}, the authors introduce the problem of placing sliding cameras on the boundary of the polygon for mobile coverage. Unlike the standard pin-hole camera model, a sliding camera in \cite{Durocher:2017} can see a point $p_1$ inside the environment if there exists a point $p_2$ along its trajectory such that the segment $p_1p_2$ is perpendicular to the trajectory of the guard and $p_1p_2$ is fully contained in the environment \cite{Durocher:2013}. The objective in \cite{Durocher:2017} is to find the minimum number of sliding cameras that can provide mobile coverage, referred to as the minimum sliding camera (MSC) problem. This problem is a variant of the art gallery problem which is NP-hard even for orthogonal polygons \cite{Schuchardt:1995}. However, it has been proven to be APX-hard on simple polygons \cite{Eidenbenz:1998}. Polynomial time approximation algorithms have been developed for orthogonal polygons \cite{Durocher:2017,Durocher:2014,Durocher:2015}. In \cite{Durocher:2013}, it is shown that MSC is NP-hard for non-simple polygons. For monotone orthogonal polygons, a linear time solution to MSC is presented in \cite{Deberg:2017}. In \cite{Durocher:2013}, the minimum length sliding camera problem is introduced. The objective is to find a set of sliding cameras that minimizes the total length of the trajectory while providing mobile coverage \cite{Almahmud:2016}. A polynomial time solution for orthogonal polygons with holes is presented in \cite{Durocher:2013}. Some variants include the concept of cameras with $k$-transmitters, introduced in \cite{Mahdavi:2014}, which allows the cameras to see through $k$ boundary walls of the polygon. In \cite{Mahdavi:2014}, the problem is shown to be NP-complete, and the authors present a $2$-approximation algorithm.

\section{Problem Statement}
\label{sec:prob}

Let $P$ be a simple $n$-vertex polygon representing a closed polygonal environment. First, we consider the case of a single unpredictable target $I$, referred to as \textit{intruder}, inside $P$. In Section \ref{sec:multi}, we extend our analysis to address multiple intruders inside the polygon. Let $\bar{v}_e$ and $x_I$ denote the maximum speed of the intruder and its location respectively at time $t$. In order to track the mobile intruder, a team of mobile observers, referred to as {\it guards}, is deployed inside the polygon. Let $\mathbb{G}$ denote the set of all guards. Each guard is equipped with an omni-directional camera with infinite range. The objective of the guards is to ensure that $I$ is visible to at least one guard at all times. The guards have knowledge of $\bar{v}_e$. Each guard has a maximum speed $\bar{v}_g$ and the \textit{speed ratio} between the guards and the intruders is defined as $r= \bar{v}_g/\bar{v}_e$.


Next, we define some graph-theoretic concepts associated with a polygon and its partitioning. A {\it triangulation} of $P$ is defined as a partition of $P$ into a set of disjoint triangles, such that the vertices of the triangles are vertices of the polygon. In general, the triangulation of a polygon is non-unique. The edges of the triangles of the triangulation are called \textit{diagonals} \cite{O'Rourke:1983}, and they can be segments inside $P$ (internal diagonals) or edges of $P$. Thus, the triangulation of $P$ is trivially represented as a planar graph $G=G(P)$ called \textit{triangulation graph}. We assume that a triangulation of $P$ is given. Let $\mathbb{V}(G)$ denote the vertex set of $G$ (corresponds to the vertices of $P$), $\mathbb{E}(G)$ denote the edge set of $G$ (corresponds to the diagonals of the triangulation of $P$), and $\mathbb{T}(G)$ (triangle set of $P$) denote the faces of $G$. Clearly, there is a bijection between the set of vertices of $P$ and $\mathbb{V}(G)$. Also, there is a bijection between the set of diagonals of the triangulation of $P$ and $\mathbb{E}(G)$. Hence, we do not make any distinction between the following pairs: (i) vertices of $P$ and the vertices in $\mathbb{V}(G)$ (ii) diagonals of the triangulation of $P$ and the edges in $\mathbb{E}(G)$ (iii) triangles of the triangulation of $P$ and the faces in $\mathbb{T}(G)$. Let $G_D$ be the dual graph of $G$. Each vertex in $G_D$ corresponds to a face in $\mathbb{T}(G)$. An edge exists between a pair of vertices in $G_D$ if the triangles in $\mathbb{T}(G)$ that correspond to such pair of vertices share an edge. Since $P$ is simply connected, $G_D$ is a tree.

The objective is to keep the intruder in the line-of-sight of at least one guard at all times, which is ensured if the triangle in which the intruder is located is covered\footnote{In this work, we define a triangle to be \textit{covered} if and only if there is a guard located at the boundary of the triangle.} by at least one guard. Each $g_i \in \mathbb{G}$ is a {\it diagonal guard} i.e., it is constrained to move along a unique diagonal $h_i \in \mathbb{H}$, where $\mathbb{H} \subset \mathbb{E}(G)$. Let $l_i$ be the length of $h_i \in \mathbb{H}$. The endpoints of $h_i$ are denoted as $v_1(i)$ and $v_2(i)$. Guards $g_i$ and $g_k$ are \textit{neighboring guards} if $g_i$ and $g_k$ are incident\footnote{We say that a guard $g_i$ is incident to triangle $T$ if at least one of the endpoints of $h_i$ is a vertex of $T$.} to the same triangle. We define $\mathbb{G}(T_k) \subseteq \mathbb{G}$ as the set of guards incident to $T_k$. 


In this paper, the distance between two points $x,y \in P$ is defined as the length of the shortest path between $x$ and $y$ on the visibility graph\footnote{A visibility graph of a polygon is a graph whose nodes corresponds to the vertices of the polygon, and there is an edge between any two vertices if the segment joining them is contained inside the polygon.} constructed using the vertices of $P$ (and the vertices of internal obstacles), $x$ and $y$. It is denoted as $d(x,y)$. The distance between $x\in P$ and a set $R_1 \subset P$ is defined as $d(x,R_1)=\min \{d(x,p) : p \in R_1 \}$. The distance between two sets of points $R_1,R_2\subset P$ is defined as $d(R_1,R_2)= \min \{d(q,R_1) : q \in R_2 \}$.

\vspace{0.1in}
\noindent
{\bf Note: Appendix \ref{sec:appc} contains a list of important variables, and their definitions. }
\vspace{0.1in}

In the next section, we describe the deployment of the guards inside the polygon.

\subsection{Selection of Dominating Diagonals}
\label{subsec:dep}
In \cite{O'Rourke:1983}, it is shown that at most $ \lfloor n/4 \rfloor$ diagonals are sufficient to dominate\footnote{A triangulation graph $G$ is said to be \textit{dominated by a set of diagonals} $\mathbb{H} \subset \mathbb{E}(G)$ if at least one vertex of each triangle in $\mathbb{T}(G)$ is an endpoint of a diagonal in $\mathbb{H}$.} $G$. Hence, $\lfloor n/4 \rfloor$ diagonal guards are sufficient to provide mobile coverage. In this section, we describe the strategy proposed in \cite{O'Rourke:1987} to identify at most $\lfloor n/4 \rfloor$ \textit{dominating diagonals}  \footnote{A set of dominating diagonals is any set of diagonals that dominate a triangulation graph.} of a polygon's triangulation. In \cite{O'Rourke:1987}, it is shown that there exists a set of dominating diagonals of size $\lfloor n/4 \rfloor$ in every triangulation graph of a polygon with $n \geq 5$. The correctness of the strategy is based on the following results \cite{O'Rourke:1987}: (i) For any triangulation graph $G$ of a simple polygon $P$ with $n \geq 10$ edges, it is always possible to find a diagonal $d \in \mathbb{E}(G)$ that partitions $G$ into subgraphs $G_1$ and $G_2$ such that $G_1$ is the triangulation graph of a hexagon, heptagon, octagon or nonagon (we call these \textit{basic polygons}). (ii) Any triangulation graph of a heptagon (or any polygon with fewer sides) has one dominating diagonal, and any triangulation graph of an octagon or nonagon has two dominating diagonals.


Algorithm \ref{alg:guards} recursively partitions $G$ into triangulation subgraphs of basic polygons. At each iteration, the algorithm searches for a diagonal $d$ that separates a triangulation subgraph of a basic polygon (denoted as $G_p$) such that there is no other diagonal in $\mathbb{E}(G_p)$ (edge set of $G_p$) that can separate a subgraph of a smaller basic polygon. Additionally, for each subgraph $G_p$ obtained, a minimal set of diagonals $\mathbb{E}(G_p)$ is found such that the diagonals in the set can dominate all the triangles\footnote{A triangle is said to be \textit{dominated} if at least one of its vertices is an endpoint of a diagonal in $\mathbb{H}$.} in $\mathbb{T}(G_p)$ (set of faces of $G_p$) that are still not dominated. The process is repeated until the remaining non-partitioned subgraph has $9$ vertices or less. The first \textbf{while} cycle (Line $5$) finds a diagonal $d$ that partitions $G$ into a pair of triangulation subgraphs, such that one of those subgraphs ($G_p$) corresponds to a basic polygon and such that there is no other diagonal in $\mathbb{E}(G_p)$ (edge set of $G_p$) that can separate a subgraph of a smaller basic polygon. This can be completed in $O(n)$ time by traversing $G_D$ (dual graph of $G$).

$G_{pol}$ (Line $10$) is a tree such that each $v \in \mathbb{V}(G_{pol})$ is associated with each subgraph $G_p$ found in Line $6$, and one vertex corresponds to the remaining subgraph $G'$ after Line $9$. An edge $e \in \mathbb{E}(G_{pol})$ exists between vertices that correspond to a pair of triangulation subgraphs $G_p$ that share a common diagonal $d \in \mathbb{E}(G)$. In the second \textbf{while} loop the minimum set of diagonals that can cover the triangles in $\mathbb{T}(G_i)$ that are not dominated by other diagonals is found (set of \textit{appropriate} diagonals, Line $14$). The second \textbf{While} loop takes $O(n)$ time by using $G_{pol}$. Hence, Algorithm \ref{alg:guards} takes $O(n)$ time.

\begin{algorithm}
	\caption{Guard Deployment}
	\begin{algorithmic}[1]
		\State\textbf{Input}: $G$
		\State\textbf{Output}: $\mathbb{H}$
		\State $S_D \leftarrow \emptyset$ is the set of diagonals $d$
		\State $G' \leftarrow G$
		\While{$G'$ has $n \geq 10$ vertices}
		\State find $d$ that separates a triangulation subgraph $G_p$ from $G'$ such that there is no triangulation subgraph of $G_p$ that corresponds to a smaller basic polygon
		\State $G'$ becomes the subgraph obtained by removing $G_p$ excepting $d$ and its vertices
		\State add $d$ to $S_D$
		\EndWhile
		\State create $G_{pol}$ from $G$ using the diagonals in $S_D$ (all vertices in $\mathbb{V}(G_{pol})$ are unmarked)
		\While{there is an unmarked vertex in $G_{pol}$}
		\State $v_i \leftarrow$ unmarked vertex in $G_{pol}$ with at most one unmarked neighbor
		\State $G_i \leftarrow$ subgraph of $G$ that corresponds to $v_i$ 
		\State $H_{i} \leftarrow$ \textit{appropriate} diagonals of $\mathbb{E}(G_i)$
		\State add diagonals of $H_{i}$ to $\mathbb{H}$
		\State mark $v_i$
		\EndWhile
	\end{algorithmic}
	\label{alg:guards}
\end{algorithm}

We call $\mathbb{H}$ a \textit{deployment} of the guards, since each $g_i \in \mathbb{G}$ is assigned to a diagonal $h_i \in \mathbb{H}$ (it is deployed along the diagonal). In the next section, we address the problem of finding the minimum speed of the guards that can ensure tracking.


\section{A Multi-robot Task Allocation Problem}
\label{sec:alloc}

A sufficient condition to track the intruder is to cover at all times the triangle in which it lies. In the previous section, we described an algorithm to select at most $ \lfloor n/4 \rfloor$ diagonals of a polygon that can dominate the triangles of the triangulation of a polygon. Given sufficient speed, the guards assigned to the diagonals in $\mathbb{H}$ can move back and forth on their diagonals to cover the triangle in which the mobile intruder lies. This gives rise to a MRTA problem in which the task for each robot is to cover the triangles allocated to it whenever the intruder lies inside them. 

First, we consider the case in which a single guard is allocated to a triangle of the triangulation. An allocation $A:{\mathbb{G}}\rightarrow2^{\mathbb{T}(G)}$ maps each guard to a subset of triangles in the triangulation of the polygon. Let $\mathcal{A}$ denote the set of all possible allocations. We say that an allocation is \textit{complete} if there is no triangle with no guard allocated to it. Next, we classify the triangles of the triangulation of $P$ (the triangles in $\mathbb{T}(G)$) based on the number of incident guards and their position relative to the triangles. Refer to Figure \ref{fig:triangles}. The red segments in the figure represent the diagonals $h_i \in \mathbb{H}$ on which the guards move. In the subsequent figures, those diagonals are labeled as $g_i$ instead.

\begin{enumerate}
	\item Safe Triangle: A triangle $T_k \in \mathbb{T}(G)$ is called \textit{safe} if $g_i$ covers $T_k$ at all times regardless of the location of $x_I$. Clearly, if there is a guard $g_i \in \mathbb{G}$ deployed on one of the edges of $T_k$, the triangle is a safe one. We use $\mathbb{T}^{safe}(G)$ to denote the set of safe triangles. In Figure \ref{fig:triangles}, safe triangles are shaded in blue.
	\item Unsafe Triangle: A triangle $T_k \in \mathbb{T}(G) \backslash \mathbb{T}^{safe}(G)$ is called \textit{unsafe} if $|\mathbb{G}(T_k)|=1$, where $|\cdot|$ is the cardinality operator. In Figure \ref{fig:triangles}, unsafe triangles are shaded in orange.
	\item Regular Triangle: A triangle $T_k$ is \textit{regular} if it is neither safe nor unsafe. In Figure \ref{fig:triangles}, regular triangles are shaded in white.
\end{enumerate}
\begin{figure}[htb]
	\begin{center} 
		\includegraphics[width=0.62\linewidth,height=0.43\linewidth]{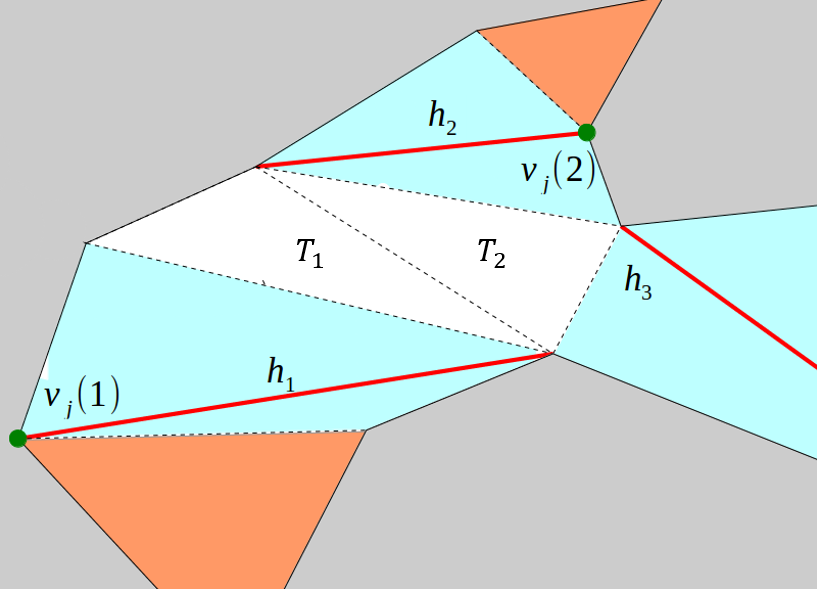}
	\end{center}
	\caption{Classification of a set of triangles in $\mathbb{T}(G)$. Notice that the regular triangles $T_1$ and $T_2$ can be covered by more than one guard.}
	\label{fig:triangles}
\end{figure}     

Let $\mathbb{T}_{\alpha}^{safe}(g_i)$ denote the set of triangles that are incident to the endpoint $v_{\alpha}(i)$ of $h_i$ ($\alpha \in \{ 1,2 \}$). Also, let $\overline{\mathbb{T}^{safe}}(G)$ $(=\mathbb{T}(G)\backslash\mathbb{T}^{safe}(G))$ denote the set of non-safe (Regular$+$Unsafe) triangles in $\mathbb{T}(G)$ and $\overline{\mathbb{T}_{\alpha}^{safe}}(g_i)$ denote the set of non-safe triangles that are incident to $v_{\alpha}(i)$. By definition, $\mathbb{T}_{\alpha}^{safe}(g_i)$ is always covered by $g_i$ for every allocation $A(\mathbb{G})\in\mathcal{A}$. We want to assign a $g_i \in \mathbb{G}$ to every $T \in \overline{\mathbb{T}^{safe}}(G)$ so that $g_i$ can cover $T$ when the intruder lies in it. Clearly, only the guards in $\mathbb{G}(T)$ can be assigned to $T$. The deployment strategy proposed in Section \ref{subsec:dep} guarantees that $|\mathbb{G}(T)| \geq 1$ for every $T \in \overline{\mathbb{T}^{safe}}(G)$. Therefore, an allocation that assigns a guard to each triangle in $\mathbb{T}(G)$ exists. Moreover, every unsafe triangle can only be assigned to the single guard incident to it.   


Given an allocation $A(\mathbb{G})$, there are two possibilities for each $g_i \in \mathbb{G}(T)$: 
\begin{enumerate}
	\item    All the non-safe triangles assigned to $g_i$ are in $\overline{\mathbb{T}_{\alpha}^{safe}}(g_i)$. In this case, $g_i$ is static, and can cover all the triangles assigned to it from $v_{\alpha}(i)$ regardless of the speed ratio $r$. 
	\item There are triangles assigned to $g_i$ in $\overline{\mathbb{T}_1^{safe}}(g_i)$ and  $\overline{\mathbb{T}_2^{safe}}(g_i)$. This implies that $g_i$ needs to be at $v_1(i)$ ($v_2(i)$) when the intruder $I$ lies in $\bigcup_{T \in \overline{\mathbb{T}_1^{safe}}(g_i)}{T}$ ($\bigcup_{T \in \overline{\mathbb{T}_2^{safe}}(g_i)}{T}$).  Therefore, $g_i$ has to move from one endpoint of $h_i$ to the other depending on $x_I$. Assume that $I$ is initially located in any triangle $T_j\in \overline{\mathbb{T}_1^{safe}}(g_i)$ assigned to $g_i$, and $I$ moves to a triangle $T_k\in \overline{\mathbb{T}_2^{safe}}(g_i)$ also assigned to $g_i$. As a result, the guard $g_i$, initially located at $v_1(i)$, moves to $v_2(i)$. Therefore, $g_{i}$ should reach $v_2(i)$ before the intruder can reach $T_k$. If the aforementioned condition is satisfied for every $T_j  \in \overline{\mathbb{T}_1^{safe}}(g_i)$ and every $T_k  \in \overline{\mathbb{T}_2^{safe}}(g_i)$ assigned to $g_i$, then $g_i$ can track $I$ when it is inside any triangle allocated to it.
\end{enumerate}

To summarize the above discussion, an allocation $A(\mathbb{G})$ should satisfy the following conditions:
\begin{enumerate}
	\item $A$ is complete.
	\item Unsafe triangles are allocated to the single guard incident to them.
	\item $\bar{v}_g\geq\max_i \bar{v}_e l_i/d(T_1^{alloc},T_2^{alloc})$, where $T_1^{alloc} \subset \overline{\mathbb{T}_1^{safe}}(g_i), T_2^{alloc} \subset  \overline{\mathbb{T}_2^{safe}}(g_i)$ are the sets of triangles allocated to $g_i$ that correspond to $v_1(i)$ and $v_2(i)$ respectively. 
\end{enumerate}
An allocation $A$ is called {\it feasible} if it satisfies the above conditions. Given the constraint that a single guard can be allocated to a non-safe triangle, we pose the following problem:



\vspace{0.1in}
\noindent
{\bf Problem 1:} Given $\bar{v}_e$, what is the minimum value of $\bar{v}_g$ and the corresponding allocation $A$ for which the intruder can be persistently tracked?

\vspace{0.1in}
$G_D$ encodes the adjacency between faces of $G$ based on common diagonals. Analogous to a dual graph, we define a {\it Guard Adjacency Graph} (GAG), denoted as $G^{\#}$, which encodes the adjacency between faces of the triangulation graph $G$ based on common guards. For a given triangulation graph $G$ of the polygon and deployment of the guards, $G^{\#}$ is constructed as follows. Each vertex of $G^{\#}$ corresponds to an non-safe triangle in $G$. An edge $e_{j,k}(g_i)$ exists between vertices $v_j$ and $v_k$ in $G^{\#}$ if there exists a guard $g_i$ such that $T_j \in \overline{\mathbb{T}_1^{safe}}(g_i)$ and $T_k \in \overline{\mathbb{T}_2^{safe}}(g_i)$. The weight of the edge $e_{j,k}(g_i)$ is given by $w_{j,k}(g_i)=l_{i}/d(T_j,T_k)$.


Next, we define terms related to allocation of guards to triangles. Given an allocation $A(\mathbb{G}) \in\mathcal{A}$, let $A(g_i)$ be the set of triangles assigned to $g_i$. Then we define the cost of the allocation of $g_i$ as $c(A(g_i))=\displaystyle\max_{T_j,T_k \in A(g_i)} \{w_{j,k}(g_i)\}$. The overall cost of the allocation for $G$ is defined as $c(A(\mathbb{G}))=\displaystyle\max_{g_i\in \mathbb{G}} \{c(A(g_i))\}$. Based on these definitions, Problem $1$ can be formulated as the following problem on $G$:\\\\
{\bf UNIALLOC}: Find a feasible allocation $A(\mathbb{G})\in\mathcal{A}$ for which $c(A(\mathbb{G}))$ is minimized.
\vspace{0.1in}

In order to prove the above theorem, first we define a set of representatives from $G^{\#}$. For each non-safe triangle $T_j \in \overline{\mathbb{T}^{safe}}(G)$, we define a set $S_g(T_j)=\{(j,g_i): g_i \in \mathbb{G}(T_j) \}$. $S_g(T_1),\ldots,S_g(T_{|\mathbb{V}(G_1)|-1}),S_g(T_{|\mathbb{V}(G^{\#})|})$ is a collection of disjoint sets. For any $x=(j,g_a) \in S_g(T_j)$ and $y=(k,g_b) \in S_g(T_k)$ with $j \neq k$ and guards $g_a,g_b$, define $c_2(x,y)$ as follows:

\begin{equation}
\label{eq:cost}
c_2(x,y)=\left \{ \begin{array}{cc}
w_{j,k}(g_a), &  a=b\\
0, & \mbox{ otherwise.}
\end{array} \right.
\end{equation}
$c_2(x,y)=0$ for $a \neq b$ models the fact that $g_a$ and $g_b$ can cover $T_j$ and $T_k$, respectively, by staying static at the corresponding endpoints. If $a=b$ and $e_{j,k}(g_a) \in \mathbb{E}(G^{\#})$, $T_j$ and $T_k$ are are located at opposite ends of $h_a$. In this case, $g_a$ should have a minimum speed of $c_2(x,y) =l_{i}/d(T_j,T_k)$ to cover triangles $T_j$ and $T_k$ if the intruder moves between them along the shortest path between them. Finally, $a=b$ and $e_{j,k}(g_a) \notin \mathbb{E}(G^{\#})$ implies that $T_j$ and $T_k$ are incident to the same endpoint of $h_a$. Since $g_a$ can cover both triangles from that endpoint, it can remain static. Therefore, $c_2(x,y) =0$ in that case.

We say that a set $S_{rep} \subset \displaystyle\bigcup_{j\in\{1,\ldots, |\mathbb{V}(G^{\#})|\}} S_g(T_j)$ is a \textit{set of representatives} if $|S_{rep} \cap S_g(T_j)|=1$ for all $j$. Based on the definition of a set of representatives, we define the following problem which is equivalent to UNIALLOC:

{\bf MAXREP}: Find the set $S_{rep}$ for which $\max \{ c_2(x,y): x,y \in S_{rep} \mbox{ and } x \neq y \}$ is minimized.

The following problem casts MAXREP as a decision problem:

{\bf GAMMAREP:} Given $\gamma \in \mathbb{R}_{\geq 0}$, does there exist a set of representatives $S_{rep}$ of size $|\mathbb{V}(G^{\#})|$ such that $c_2(x,y) \leq \gamma$ for all $x,y \in S_{rep}$ with $x \neq y$?

The next lemma proves that GAMMAREP is NP-hard which in turn implies that MAXREP is at least NP-hard. 


\begin{lemma}
	\label{lemma:gamma}
	GAMMAREP is NP-hard.
\end{lemma}
\begin{proof}
	We reduce the problem of finding a $K$-clique \textbf{\cite{Karp:1972}} in a graph (a NP-hard problem \cite{Chen:2005}) to GAMMAREP. Consider a graph without self-loops $G_2$ such that $|\mathbb{V}(G_2)| > |\mathbb{V}(G^{\#})|$. The problem of finding a $|\mathbb{V}(G^{\#})|$-clique in $G_2$ is NP-hard \cite{Garey:2002}. For each $v_j \in \mathbb{V}(G_2)$, we define a set $S_j=\{ j \} \times \mathbb{V}(G_2)$. For a pair of vertices $v,w \in \mathbb{V}(G_2)$, we define the following cost function:
	
	\begin{equation}
	\label{eq:cost2}
	c_3((j,v),(k,w))= \left\{ \begin{array}{cc} 0, & \mbox{ if an edge exists between $v$ and}\\& \mbox{$w$ in $G_2$} \\ 1, & \mbox{ otherwise} \end{array} \right.
	\end{equation}
	
	
	Given sets $S_j$ and the cost function $c_3$ stated above, we prove that a $|\mathbb{V}(G^{\#})|$-clique in $G_2$ exists if and only if a set of $|\mathbb{V}(G^{\#})|$ representatives $S_{rep} \subset \displaystyle\bigcup_{j\in\{1,\ldots, |\mathbb{V}(G^{\#})|\}} S_j$ such that $c_3((j,v),(k,w)) = 0 \text{ }\forall \text{ }(j,v),(k,w) \in S_{rep}$ can be found. 
	
	($\Rightarrow)$ First, assume that there exists such a $S_{rep}$ but there is no $|\mathbb{V}(G^{\#})|$-clique in $G_2$. By definition, for each pair $(j,v),(k,w) \in S_{rep}$, $j \neq k$, and $v \neq w$ ($v=w$ implies that $G_2$ has a self-loop which is not possible). Therefore, each $(j,v) \in S_{rep}$ is associated to a different $v \in \mathbb{V}(G_2)$. Since there is no $|\mathbb{V}(G^{\#})|$-clique in $G_2$, the subgraph induced by the vertices associated with $S_{rep}$ is not complete. Therefore, there is at least one pair $(j,v),(k,w) \in S_{rep}$ such that there is no edge shared between $v$ and $w$ in $G_2$ which implies $c_3((j,v),(k,w)) \neq 0$ (a contradiction). Therefore, if there exists $S_{rep}$ of size $|\mathbb{V}(G^{\#})|$ such that $c_3((j,v),(k,w)) = 0 $ for all $(j,v),(k,w) \in S_{rep}$ then a $|\mathbb{V}(G^{\#})|$-clique in $G_2$ exists. 
	
	($\Leftarrow$) Now assume that there is no $S_{rep}$ which meets the aforementioned constraints, but there is a $|\mathbb{V}(G^{\#})|$-clique in $G_2$. Since $G_2$ contains a $|\mathbb{V}(G^{\#})|$-clique, there is a subset of $|\mathbb{V}(G^{\#})|$ vertices such that its induced subgraph $\hat{G}_2$ is complete. Since $\hat{G}_2$ is complete, choosing $(j,v_j) \in S_j$ for each $v_j \in \mathbb{V}(\hat{G}_2)$ as an element of $S_{rep}$ implies that for any $(j,v_j),(k,v_k) \in S_{rep}$, an edge exists between $v_j$ and $v_k$ in $G_2$. Therefore, $c_3(v_j,v_k) =0$ for every $(j,v_j),(k,v_k) \in S_{rep}$ which contradicts the assumption that there is no $S_{rep}$ of size $|\mathbb{V}(G^{\#})|$ such that $c_3((j,v_j),(k,v_k)) = 0 $ for all $(j,v_j),(k,v_k) \in S_{rep}$. Therefore, if there exists a $|\mathbb{V}(G^{\#})|$-clique in $G_2$ then there is $S_{rep}$ of size $|\mathbb{V}(G^{\#})|$ such that $c_3((j,v),(k,w)) = 0 $ for all $(j,v),(k,w) \in S_{rep}$, and the reduction can be done in polynomial time. Thus, we have a polynomial-time reduction from $K$-clique to GAMMAREP.
\end{proof}

Based on the above discussion, we can state the following theorem.
\noindent
\begin{theorem}
	\label{theorem:unialloc}
	UNIALLOC is NP-hard.
\end{theorem}
\begin{proof}
	The proof follows from the fact UNIALLOC is equivalent to MAXREP which is at least NP-hard.
\end{proof}
\subsection{Suboptimal Algorithm for Computing Approximate minimum speed ratio}
\label{subsec:subop}
MAXCLIQUE refers to the problem of finding the maximum-sized clique in a graph. MAXCLIQUE is NP-hard \cite{Pardalos:1994}. An approximation algorithm for MAXCLIQUE can be used to obtain a suboptimal solution for Problem 1. The procedure is as follows. Let $G_3$ be a graph such that $\mathbb{V}(G_3)=\{v_j^i:  \exists S_g(T_j) \mbox{ such that } (j,g_i) \in S_g(T_j) \}$. Recall that $S_g(T_j)=\{(j,g_i): g_i \in \mathbb{G}(T_j) \}$ for all $1 \leq j \leq |\mathbb{V}(G^{\#})|$. Since $|\mathbb{V}(G^{\#})|=O(n)$ ($|\mathbb{V}(G^{\#})| \leq n-2$), and $ 1 \leq |S_g(T_j)| \leq |\mathbb{G}|=O(n)$, $|S_g(T_j)|=O(n)$. Therefore, $|\mathbb{V}(G_3)|=O(n^2)$. Additionally, let $\mathbb{E}(G_3)=\{e_{j,k}^{a,b}: v_j^a,v_k^b \in \mathbb{V}(G_3) \mbox{ with } j \neq k  \}$. Each edge $e_{j,k}^{a,b}$ has an associated weight defined as follows:

\begin{equation}
\label{eq:cost3}
w_{j,k}^{a,b}= \left\{ \begin{array}{cc} 0, & \mbox{ if } a \neq b\mbox{ or } \frac{l_a}{d(T_j,T_k)} \leq r\\ 1, & \mbox{ otherwise} \end{array} \right.
\end{equation}

We define $G_3'$ as a subgraph of $G_3$ such that $\mathbb{V}(G_3')=\mathbb{V}(G_3)$ and $\mathbb{E}(G_3') = \{ e_{j,k}^{a,b} \in \mathbb{E}(G_3): w_{j,k}^{a,b}=0  \}$. Therefore, $\mathbb{E}(G_3')$ consists of edges in $\mathbb{E}(G_3)$ with weights less than or equal to $r$. We can use any approximation algorithm for MAXCLIQUE to find a $|\mathbb{V}(G_1)|-$clique in $G_3'$ for a given $r$. Since $r_{min}=\min \{ r=\frac{l_i}{d(T_j,T_k)}: e_{j,k}(i) \in G^{\#}  \}$, we perform the aforementioned check for values of $r$ that correspond to the weights of the edges in $G^{\#}$. The suboptimal allocation corresponds to the minimum value of $w_{j,k}(g_i)$ for which a $|\mathbb{V}(G_1)|-$clique is found in the graph $G_3'$. If no such $w_{j,k}(g_i)$ is found, $r_{min}$ is equal to the minimum $w_{j,k}(g_i)$. Each vertex $v_j^i \in \mathbb{V}(G_3')$ in the $|\mathbb{V}(G_1)|-$clique corresponds to a distinct triangle in $T_j \in \overline{\mathbb{T}^{safe}}$, and each one of them corresponds to a $g_i \in \mathbb{G}$ which is the guard assigned to $T_j$. This gives an approximate optimal allocation $A(g_i)$ for each $g_i \in \mathbb{G}$. 


It has also been shown that approximating the MAXCLIQUE within a factor of $n^\epsilon$ for some $\epsilon>0$ is NP-hard \cite{Arora:1998,Hastad:1999}. In the past, probabilistic techniques \cite{Babel:1991,Balas:1996,Wood:1997,Pardalos:1994} have been proposed in the literature to find largest cliques without any guarantees on the optimality. \cite{Feige:2004} presents an approximation algorithm that finds a clique with an approximation ratio of $O(n (\log\log n)^2 /(\log n)^3)$ when the size of the maximum clique is between $n/ \log n$ and $n / (\log n)^3$. Since $\mathbb{V}(G_3)=O(n^2)$, where $n$ is the number of vertices of the polygon, the performance ratio in our case is $O(n^2 (\log (\log n))^2 /(\log n)^3)$.


Next, we consider the case in which more than one guard can be assigned to each non-safe triangle. In this case, we partition $T_j$ into disjoint regions, each covered by a unique guard (incident on $T_j$). The problem of computing the minimum guard speed ($\bar{v}_g$) when multiple guards can be assigned to a single triangle is as hard as UNIALLOC which itself is NP-hard. 

In the next section, we address the problem of allocating guards to triangles for a given maximum speed of the intruder and the guards.

\section{Computing a Feasible Allocation for Known Guard Speed}
\label{sec:gen}



In the previous section, we proved that the problem of finding the minimum speed of the guards that can ensure tracking is NP-hard. In this section, we propose a technique to search and compute a feasible allocation for given maximum speed of the intruder and guards. Therefore, the speed ratio $r$ is known. We address the general allocation problem in which multiple guards can be assigned to cover an unsafe triangle in the polygon. Specifically, we focus on techniques that partition the polygon, and activate guards to track an intruder located in its allocated partition. 

\vspace{0.1in}
Let $R_j^\alpha(i)$ denote the region inside triangle $T_j$ assigned to the guard $g_i \in \mathbb{G}(T_j)$ incident to vertex $v_{\alpha}(i)$ ($\alpha = \{1,2\}$). We define $S_R^{\alpha}(i)=\{ R_j^\alpha(i) | R_j^\alpha(i) \subseteq T_j\text{ }\in \overline{\mathbb{T}_\alpha^{safe}}(g_i) \}$ as the set of regions assigned to $g_i$ that can be covered from $v_{\alpha}(i)$. We define $\hat{U}_R^{\alpha}(i)= \displaystyle\bigcup_{R_j^{\alpha}(i) \in S_R^{\alpha}(i)} R_j^{\alpha}(i)$ as the union of the regions belonging to $S_R^{\alpha}(i)$. Let $\mathbb{R}^{alloc}(g_i)=\hat{U}_R^{1}(i) \cup \hat{U}_R^{2}(i)$ denote the region assigned to $g_i$. The cost associated with an assignment is defined as $c(\mathbb{R}^{alloc}(g_i))= l_{i}/d(\hat{U}_R^{1}(i),\hat{U}_R^{2}(i))$.

In this section, we address the following problem: 

\vspace{0.1in}
\noindent
{\bf Problem 2:} For the $\lfloor n/4 \rfloor$ deployment of guards inside a polygon described in Section \ref{subsec:dep} and a given $r$, determine $\mathbb{R}^{alloc}(g_i)$ for each $g_i \in \mathbb{G}$ such that $\max \{c(\mathbb{R}^{alloc}(g_i)): g_i \in \mathbb{G}  \} \leq r$, and the region in which the intruder lies can be covered by the guard allocated to it. 
\vspace{0.1in}

Since $c(\mathbb{R}^{alloc}(g_i))= l_{i}/d(\hat{U}_R^{1}(i)$,$\hat{U}_R^{2}(i)) \leq r$, $d(\hat{U}_R^{1}(i),\hat{U}_R^{2}(i)) \geq l_{i}/r$ for all $g_i \in \mathbb{G}$. Let $d_I^i=l_{i}/r$ denote the maximum distance traveled by the intruder during the time in which $g_i$ can travel from one endpoint of $h_i$ to the other.

In the next section, we present the procedure to find $\mathbb{R}^{alloc}(g_i)$.

\subsection{Sequential computation of partitions}
\label{subsec:cont}
The allocation algorithm proposed in the next section sequentially computes $\hat{U}_R^{\alpha}(i)$ for each guard $g_i$. The end point of $h_i$ at which $\hat{U}_R^{\alpha}(i)$ is computed first gets the label $\alpha=1$. Subsequently, $\hat{U}_R^{\alpha}(i)$ is computed at the other end point of $h_i$ which is assigned the label $\alpha=2$.

From the definition of $\hat{U}_R^{1}(i)= \displaystyle\bigcup_{ R_j^{1}(i) \in S_R^{1}(i)} R_j^{1}(i)$, it is clear that $R_j^{1}(i)$ must be defined for every $T_j \in \overline{\mathbb{T}_{1}^{safe}}(g_i)$ for computing $\hat{U}_R^{1}(i)$. Based on the classification of $T_j \in \overline{\mathbb{T}_{1}^{safe}}(g_i)$, we can have the following scenarios: 
\begin{enumerate}
	\item If every $T_j \in \overline{\mathbb{T}_{1}^{safe}}(g_i)$ is an unsafe triangle, $R_j^{1}(i)=T_j$ by definition since $T_j$ can only be covered by $g_i$. In this case, $\hat{U}_R^{1}(i)= \displaystyle\bigcup_{T_j \in \overline{\mathbb{T}_{1}^{safe}}(g_i)} T_j$. 
	
	\item If $T_j$ is a regular triangle, $R_j^{1}(i)$ can be computed only if the region allocated to other guards incident to $T_j$ is known. In this case, $R_j^{1}(i) =\displaystyle\bigcap_{g_k \in \mathbb{G}(T_j) \backslash \{ g_i \} }{\overline{R}_j^2(k)} \cap T_j$, where $\overline{R}_j^2(k)$ is the complement of $R_j^2(k)$.
	
\end{enumerate}

The following equation allows computing $\hat{U}_R^2(i)$ from $\hat{U}_R^{1}(i)$:

\begin{equation}
\label{eq:r_2}
R_j^2(i)= \left\{ \begin{array}{cc}
\displaystyle\bigcap_{p \in \hat{U}_R^1(i)} \beta^{c}_{d_I^i}(p) \cap T_j^{free}, & \hat{U}_R^1(i) \neq \emptyset\\
T_j^{free}, & \mbox{otherwise}
\end{array}
\right.
\end{equation}
where $\beta_{d_I^i}(p)$ is an open ball (using the metric defined in Section \ref{sec:intro}) of radius $d_I^i$ centered at $p$, and $\beta^{c}_{d_I^i}(p)$ is its complement. Also, $T_j^{free} \subseteq T_j$ is the region inside $T_j$ that has not yet been assigned to a guard. In the absence of obstacles between $R_j^2(i)$ and $\hat{U}_R^1(i)$:
\begin{equation}
\label{eq:r_minkowski}
R_j^2(i)= (\hat{U}_R^1(i) \bigoplus B_{d_I^i})^c \cap T_k^{free},
\end{equation}
where $B_{d_I^i}$ is an open ball in $\mathbb{R}^2$, $\bigoplus$ denotes the Minkowski sum, and $(\hat{U}_R^1(i) \bigoplus \beta_{d_I^i})^c$ is the complement of the Minkowski sum. The set $(\hat{U}_R^1(i) \bigoplus B_{d_I^i})$ is called the offset of $\hat{U}_R^1(i)$. The offset of a set bounded by a polyline curve (line-segments/arcs of circle) is a set bounded by polyline curve \cite{Liu:2007}. Appendix \ref{sec:appa} shows that $R_j^2(i)$ and $\hat{U}_R^2(i)$ are sets bounded by polyline curves even in the presence of obstacles. From $\hat{U}_R^{1}(i)$ and $\hat{U}_R^{2}(i)$, we can compute $\mathbb{R}^{alloc}(g_i)$ using the equation $\mathbb{R}^{alloc}(g_i)=\hat{U}_R^{1}(i) \cup \hat{U}_R^{2}(i)$. 

In Figure \ref{fig:cont}, there are three guards, $g_1$, $g_2$, and $g_3$. All the non-safe triangles incident to one endpoint of the diagonals of $g_1$ and $g_3$ are unsafe triangles. Therefore, $\hat{U}_R^{1}(1)$ and $\hat{U}_R^{1}(3)$ are as shown in the figure. However, that is not the case for $g_2$. In order to define $\hat{U}_R^{1}(2)$, either $g_1$ or $g_3$ needs to have a region assigned to it in the triangles shared between it and $g_2$. In Figure \ref{fig:cont}, $\hat{U}_R^{2}(1)$ (which is obtained from (\ref{eq:r_2})) is used to define $\hat{U}_R^{1}(2)$, which in turn is used to define $\hat{U}_R^{2}(2)$ using (\ref{eq:r_2}). Once that all the guards in $\mathbb{G}(T_j)$ have been assigned to a region $R_j^{\alpha}(i)$, it is possible to determine if no allocation was found for the given $r$. By definition, a region $R_j^{1}(i)$ inside $T_j$ ensures that every location within the triangle is assigned to a guard in $\mathbb{G}(T_j)$. However, that may not hold if all the regions allocated to the guards in $T_j$ are $R_j^{2}(i)$. Every $R_j^{2}(i)$ is constructed using (\ref{eq:r_2}), thereby, guaranteeing that $g_i$ is assigned to the largest possible region within $T_j$ such that $c(\mathbb{R}^{alloc}(g_i)) \geq r$. Otherwise, there is no guarantee to keep track of the intruder if it follows the shortest path between $\hat{U}_R^{1}(i)$ and $\hat{U}_R^{2}(i)$. Clearly, an allocation cannot be found by the above process if there is a region inside the triangle that is not assigned to any guard in $\mathbb{G}(T_j)$ after every $R_j^{2}(i)$ is defined for $T_j$. In Figure \ref{fig:cont}, there is a region $R_{\emptyset}$ inside the regular triangles shared by $g_2$ and $g_3$ that cannot be assigned to those guards. The existence of a region $R_{\emptyset} \neq \emptyset$ implies that an allocation that guarantees successful tracking was not found. The blue shaded triangles correspond to safe triangles, which by definition are covered all the time.

\begin{figure}[htb]
	\begin{center} 
		\includegraphics[width=0.62\linewidth,height=0.43\linewidth]{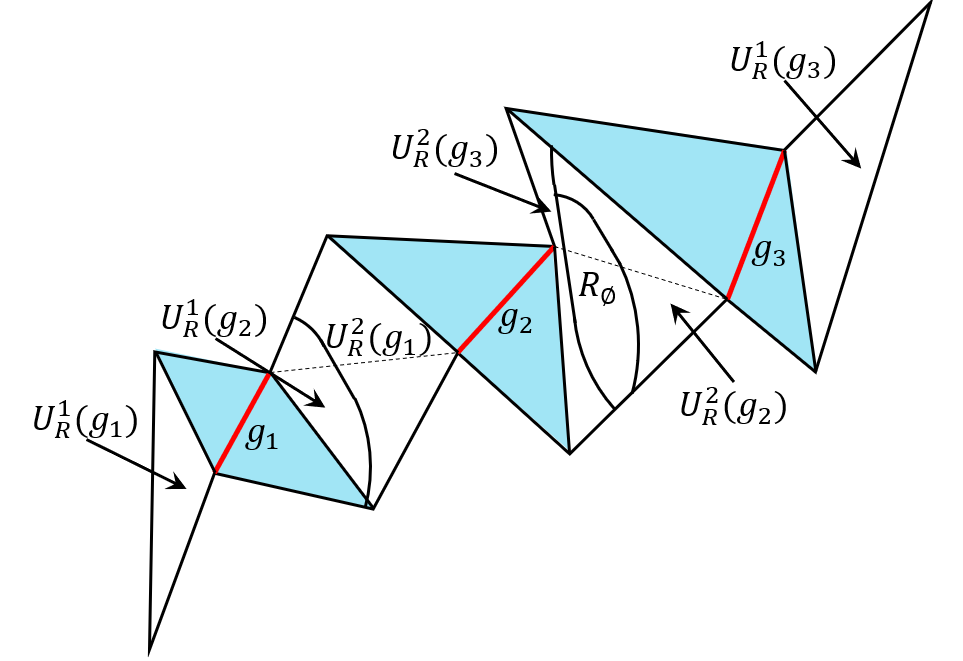}
	\end{center}
	\caption{Three guards deployed in a polygon.}
	\label{fig:cont}
\end{figure}

\subsection{Guard Allocation}
\label{subsec:disc}

In this section, we present an allocation algorithm that respects the constraints in Problem 2. Algorithm \ref{alg:allocation} presents the pseudocode of the allocation technique. Given a  triangulation graph $G$ and the speed ratio $r$ as the input, Algorithm \ref{alg:allocation} provides a partition of the polygon $P$, and an allocation of the guards to the partitions. The initialization, update and termination step of Algorithm \ref{alg:allocation} are as follows:

\begin{itemize}
	\item Initialization:
	\begin{itemize}
		\item Initialize two queues: $\mathbb{G}_{ready}=\mathbb{G}_{alloc}= \emptyset$.
		\item If there is a $g_i \in \mathbb{G}$ such that $h_i$ has an endpoint $v_{\alpha}(i)$ at which all the non-safe triangles incident to it are unsafe triangles, then $g_i$ is added to $\mathbb{G}_{ready}$. 
		\item If no guard satisfies the previous criteria, $\mathbb{G}_{ready}=\emptyset$. A guard $g_i$ is selected using Algorithm \ref{alg:arbitrary}, and all the triangles incident to one of its endpoints are assigned to it.
	\end{itemize}
	\item Update:
	\begin{itemize}	
		
		\item If there is a guard $g_i \notin \mathbb{G}_{ready} \cup \mathbb{G}_{alloc}$ such that for each $T_j \in \overline{\mathbb{T}_{\alpha}^{safe}}(g_i)$, the regions $\overline{R}_j^2(k)$ are already defined and allocated for all $g_k \in \mathbb{G}(T_j) \backslash \{ g_i \}$. $v_{\alpha}(i)$ is called the {\it preferential endpoint} of $g_i$, and is labeled as $v_{1}(i)$.
		\item A guard $g_i \in \mathbb{G}_{ready}$ is selected. Algorithm \ref{alg:orient} is used to compute $\hat{U}_R^{1}(i)$ and $\hat{U}_R^{2}(i)$.
		\item If $\mathbb{G}_{ready}= \emptyset$, a guard is arbitrarily selected using Algorithm \ref{alg:arbitrary}. Thus, Algorithm \ref{alg:allocation} can continue as in the previous step.
	\end{itemize}
	\item Termination: The algorithm can terminate in two ways described as follows.
	\begin{itemize}
		\item Algorithm \ref{alg:orient} finds a region $R_{\emptyset}(j) \neq \emptyset$ inside a non-safe triangle $T_j$. Since the region $R_{\emptyset}(j)$ cannot be assigned to any guard, the algorithm terminates. Therefore, no allocation is found.
		\item Each $T_j \in \overline{\mathbb{T}^{safe}}(G)$ is partitioned, and each partition is assigned to a guard. In this case a feasible allocation is found.
	\end{itemize}
\end{itemize}

Figure \ref{fig:ex_allocation1} (a) shows a polygonal environment with $21$ ($=n$) edges and $4$ ($\leq\lfloor n/4 \rfloor$) guards. There are $8$ safe and $11$ non-safe triangles. The corresponding graph $G^{\#}$ is shown in Figure \ref{fig:ex_allocation1} (c). The orientation of the edges of $G^{\#}$ illustrates a partial order in which the vertices of $G^{\#}$ are allocated by the algorithm. $\mathbb{V}(G^{\#})$ has $11$ vertices, each one associated with an non-safe triangle. Edges exist between vertices that correspond to triangles incident to opposite endpoints of the diagonal of a guard. For example, $T_9$ and $T_{11}$ are incident to one endpoint of $h_4$ and $T_{10}$ is incident to the other. Hence, edges in $G^{\#}$ that correspond to $g_4$ are $e_{10,9}(g_4)$ and $e_{10,11}(g_4)$. In Algorithm \ref{alg:allocation}, once a $g_i \in \mathbb{G}_{ready}$ is chosen (and $\mathbb{R}^{alloc}(g_i)$ is obtained by Algorithm \ref{alg:orient}), the endpoints of its diagonal can be labeled as $v_1(i)$ and $v_2(i)$.


In the example shown in Figure \ref{fig:ex_allocation1}, when the algorithm starts, $\mathbb{G}_{ready}= \{ g_1,g_4 \}$. $g_1$ meets the definition of the $\mathbb{G}_{ready}$ queue since all the non-safe triangles incident to one endpoint of $h_1$ are unsafe triangles ($T_1$ and $T_2$). The same is true for one endpoint of $h_4$ ($T_{10}$ is an unsafe triangle). Next, $g_1$ is selected and its corresponding preferential endpoint is named $v_1(1)$. Algorithm \ref{alg:orient} is called. It uses the procedure to compute $\mathbb{R}^{alloc}(g_i)$ from Section \ref{subsec:cont} for a given $g_i$. Thus, $R_1^1(1)=T_1$ and $R_2^1(1)=T_2$, $S_R^1(1)= \{ T_1,T_2 \}$ and $\hat{U}_R^1(1)=T_1 \cup T_2$. Edges $e_{1,3}(1)$, $e_{1,4}(1)$, $e_{1,5}(1)$ and $e_{1,8}(1)$ become outgoing edges of $v_1$, and edges $e_{2,3}(1)$, $e_{2,4}(1)$, $e_{2,5}(1)$ and $e_{2,8}(1)$ become outgoing edges of $v_2$. Next, the regions that correspond to the endpoint $v_2(1)$ are computed. They are $R_3^2(1)=T_3$, $R_5^2(1)=T_5$, $R_8^2(1)=T_8$ and $R_4^2(1)$ which is the region inside $T_4$ assigned to $g_1$ shown in Figure \ref{fig:ex_allocation1} (b). Finally, the algorithm searches for the existence of any region $R_{\emptyset}(j) \neq \emptyset$ with $j \in \{3,4,5,8\}$. Since it is not found, the process returns to Algorithm \ref{alg:allocation} where $\mathbb{G}_{alloc}$ (which contains all guards that have been allocated) and $\mathbb{G}_{ready}$ are updated to $\mathbb{G}_{alloc}=\{g_1\}$ and $\mathbb{G}_{ready}=\{g_4,g_2\}$. In the second iteration, $g_4$ is selected, and the same process is repeated. $R_{10}^1(4)=\hat{U}_R^1(4)=T_{10}$. Edges $e_{10,9}(4)$ and $e_{10,11}(4)$ become outgoing edges of edges of $v_{10}$. $R_{11}^2(4)=T_{11}$ and $R_9^2(4)$ (the region inside $T_9$ assigned to $g_4$ shown in Figure \ref{fig:ex_allocation1} (b)) are obtained. At the end of the iteration, $\mathbb{G}_{alloc}=\{g_1,g_4\}$ and $\mathbb{G}_{ready}=\{g_2,g_3\}$. The algorithm continues until the end of the fourth iteration at which point $\mathbb{G}_{alloc}=\{g_1,g_4,g_2,g_3\}$ and $\mathbb{G}_{ready}= \mathbb{G}_{\Omega} (=\mathbb{G} \backslash (\mathbb{G}_{ready} \cup \mathbb{G}_{alloc}))= \emptyset$. The allocation of all the regions of the environment is shown in Figure \ref{fig:ex_allocation1} (b), and the resulting directed graph $G^{\#}$ is illustrated in Figure \ref{fig:ex_allocation1} (c). Notice that Line $15$ of Algorithm \ref{alg:orient} is only reached when all the edges incident to a vertex $v_j$ in $G^{\#}$ are incoming edges. Thus, the regions $\hat{U}_R^1(g_i)$ and $\hat{U}_R^2(g_i)$ for each guard $g_i \in \mathbb{G}(T_j)$ are already defined. Therefore, the algorithm checks whether the triangle is completely covered by the regions allocated to the guards or there is a region inside the triangle that cannot be assigned to guards. In the example, the aforementioned check is performed for triangles $T_8$, $T_7$, $T_3$, $T_{11}$, $T_6$ and $T_5$. We can see that for all the other vertices, there are outgoing edges incident to them, and those edges are associated with a single guard. According to the definition of $\mathbb{G}_{ready}$, at each iteration Algorithm \ref{alg:allocation} selects a guard $g_i$ such that the regions of the other guards that can cover the triangles in $\overline{\mathbb{T}_{1}^{safe}}(g_i)$ are already defined. It follows that Algorithm \ref{alg:arbitrary} selects the vertices that correspond to the triangles in $\overline{\mathbb{T}_{1}^{safe}}(g_i)$ and orients all the edges incident to them that correspond to $g_i$ as outgoing edges of those vertices. Since all the other guards that can cover those triangles had their regions defined, it means that Algorithm \ref{alg:arbitrary} was called before to orient their edges so all of those edges are outgoing edges of other vertices. Hence, they are incoming edges of the vertices associated with the triangles in $\overline{\mathbb{T}_{1}^{safe}}(g_i)$. This explains why either all the edges incident to a vertex in $G^{\#}$ become incoming edges, or only the edges associated to one guard become outgoing edges.

\begin{figure}[thpb]
	\begin{center}
		\begin{subfigure}{0.9\linewidth}
			\includegraphics[width=0.75\linewidth,height=0.55\linewidth]{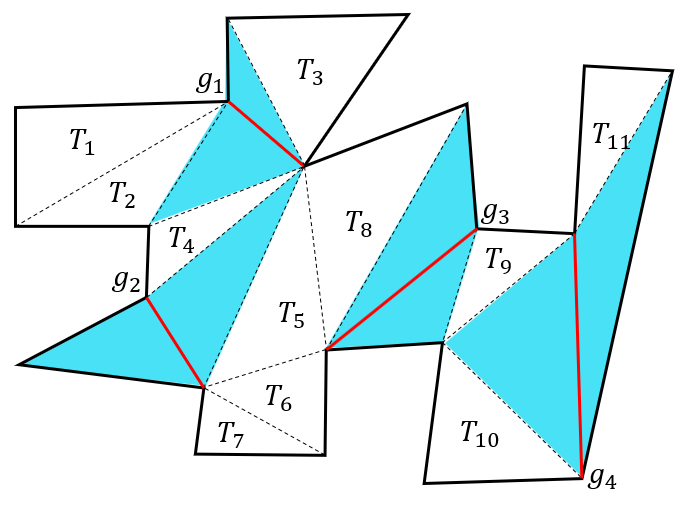}
			\caption{}
			\label{fig:ex_allocation1a}
		\end{subfigure}
		\begin{subfigure}{0.55\linewidth}
			\includegraphics[width=0.84\linewidth,height=0.55\linewidth]{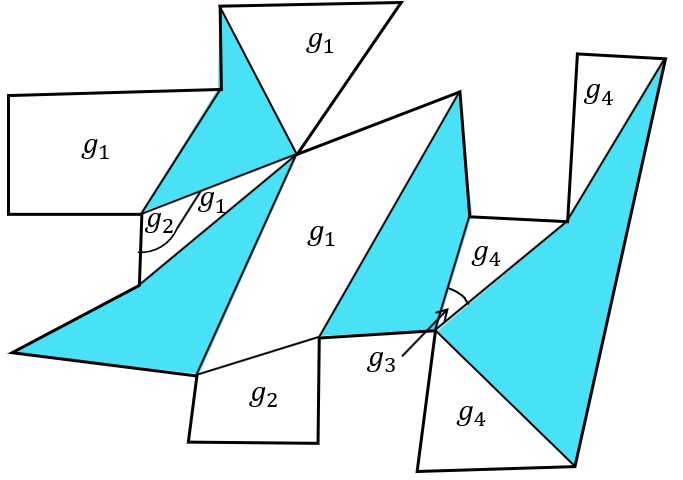}
			\caption{}
			\label{fig:ex_allocation1b}
		\end{subfigure}
		\begin{subfigure}{0.3\linewidth}
			\includegraphics[width=0.8\linewidth,height=1.6\linewidth]{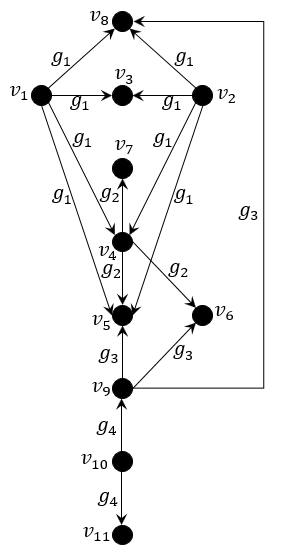}
			\caption{}
			\label{fig:ex_allocation1c}
		\end{subfigure}
	\end{center}
	\caption{(\subref{fig:ex_allocation1a}) Triangulated $P$ with guards deployed. (\subref{fig:ex_allocation1b}) The regions of the resulting partition of the non-safe triangles are assigned to the guards. (\subref{fig:ex_allocation1c}) Resulting graph $G^{\#}$ with the orientation of all edges defined.}
	\vspace{-0.1in}
	\label{fig:ex_allocation1}
\end{figure}

\begin{algorithm}
	\caption{GENALLOC}
	\begin{algorithmic}[1]
		\State\textbf{Input}:\textbf{$P$,$G$,$r$ and $\mathbb{G}$.}
		\State\textbf{Output}: Triangles of the triangulation of $P$ partitioned into regions assigned to the guards.
		
		\State \textbf{$\mathbb{V}(G^{\#}) \leftarrow \{ v_j : \exists T_j \in \overline{\mathbb{T}^{safe}}(G) \}$}
		\State \textbf{$\mathbb{E}(G^{\#}) \leftarrow \{ e_{j,k}(g_i) : \exists g_i \in \mathbb{G} \mbox{ such that } T_j \in \overline{\mathbb{T}_1^{safe}}(g_i) \mbox{ and } T_k \in \overline{\mathbb{T}_2^{safe}}(g_i)  \}$}
		\State \textbf{update $\mathbb{G}_{ready}$}	
		\While{$\mathbb{G}_{ready} \neq \emptyset$}
		\State choose $g_i \in \mathbb{G}_{ready}$
		\State $v_1(i) \leftarrow$ preferential endpoint of $h_i$
		\State $\overline{\mathbb{T}_{1}^{safe}}(g_i)$ is obtained
		\State $stop \leftarrow$ call Algorithm \ref{alg:orient}
		\If{$stop=$ True}
		\State the allocation is not possible
		\State \Return
		\EndIf
		\State update $\mathbb{G}_{ready}$ and $\mathbb{G}_{alloc}$
		\EndWhile					
		\State $\mathbb{G}_{\Omega} \leftarrow \mathbb{G} \backslash (\mathbb{G}_{ready} \cup \mathbb{G}_{alloc})$	
		\If{$\mathbb{G}_{\Omega} \neq \emptyset$}
		\State call Algorithm \ref{alg:arbitrary}
		\State update $\mathbb{G}_{ready}$ and $\mathbb{G}_{alloc}$
		\State go to $3$
		\EndIf
		\State \Return
	\end{algorithmic}
	\label{alg:allocation}
\end{algorithm}

\begin{algorithm}
	\caption{LOCALLOC}
	\begin{algorithmic}[1]
		\State\textbf{Input}: $g_i$,$\overline{\mathbb{T}_{1}^{safe}}(g_i)$,$G^{\#}$ and $\mathbb{G}$.
		\State\textbf{Output}: Regions $\hat{U}_R^1(i)$ and $\hat{U}_R^2(i)$ are defined and assigned to $g_i$ or no feasible allocation found.			
		\State $\mathbb{E}(i) \leftarrow \emptyset$	
		\For{each $T_j \in \overline{\mathbb{T}_{1}^{safe}}(g_i)$}
		\State $R_j^1(i) \leftarrow  \bigcap_{g_k \in \mathbb{G}(T_j) \backslash \{ g_i \} }{\overline{R}_j^2(k)} \cap T_j$
		\State $\mathbb{E}^j(i) \leftarrow $ edges incident to $v_j$ that correspond to $g_i$
		\State orient edges $e_{j,\cdot}(i) \in \mathbb{E}^j(i)$ as outgoing edges of $v_j$
		\State add edges in $\mathbb{E}^j(i)$ to $\mathbb{E}(i)$
		\EndFor
		\State compute $S_R^{1}(i)$ and $\hat{U}_R^1(i)$
		\For{each $e \in \mathbb{E}(i)$}
		\State $v_j \leftarrow$ head of $e$				
		\State compute $R_j^2(i)$ using (\ref{eq:r_2})
		\If{for all $g_k \in \mathbb{G}(T_j)$, $R_j^2(i)$ is defined}
		\State $R_{\emptyset}(j) \leftarrow \bigcap_{g_k \in \mathbb{G}(T_j)}{\overline{R}_j^2(k)} \cap T_j$
		\If {$R_{\emptyset}(j) \neq \emptyset$}
		\State \Return True
		\EndIf
		\EndIf
		\EndFor	
		\State compute $S_R^{2}(i)$ and $\hat{U}_R^2(i)$
		\State \Return False
	\end{algorithmic}
	\label{alg:orient}
\end{algorithm}

If $\mathbb{G}_{ready}=\emptyset$ and $\mathbb{G}_{\Omega}\neq \emptyset$ (Line $18$), Algorithm \ref{alg:allocation} calls Algorithm \ref{alg:arbitrary} to allocate the guards in $\mathbb{G}_{\Omega}$. It happens when the iterative procedure cannot find a unique allocation for the non-safe triangles that can be covered by the guards in the set $\mathbb{G}_{\Omega}$. Let $\mathbb{T}_{\Omega}$ denote the triangles that can be covered by the guards in $\mathbb{G}_{\Omega}$ such that they contain regions that have not yet been assigned. Let $\mathbb{V}_{\Omega}$ be the set of vertices in $\mathbb{V}(G^{\#})$ that correspond to the triangles in $\mathbb{T}_{\Omega}$. Triangles in $\mathbb{T}_{\Omega}$ contain regions that can be assigned to more than one guard, and hence, it is not possible to determine a unique partition. Thus, for each $v_j \in \mathbb{V}_{\Omega}$ there are at least two guards with non-oriented edges incident to $v_j$. 


\begin{figure}[thpb]
	\begin{center}
		\begin{subfigure}{0.64\linewidth}
			\includegraphics[width=0.9\linewidth,height=0.65\linewidth]{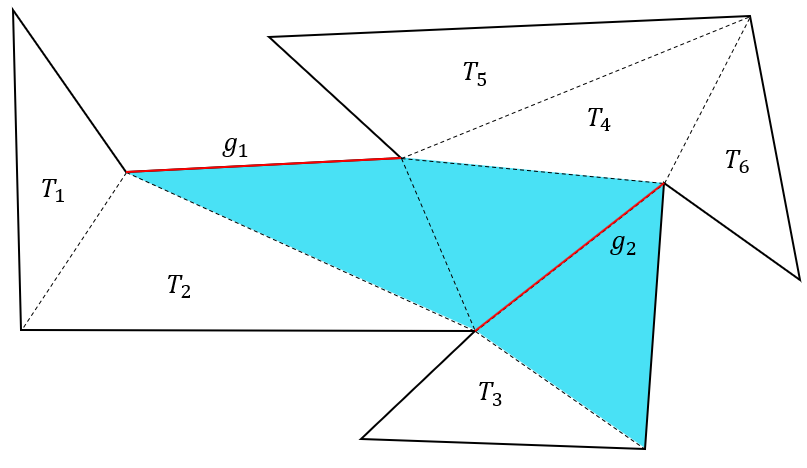}
			\caption{}
			\label{fig:ex_allocation2a}
		\end{subfigure}
		\begin{subfigure}{0.59\linewidth}
			\includegraphics[width=1\linewidth,height=0.8\linewidth]{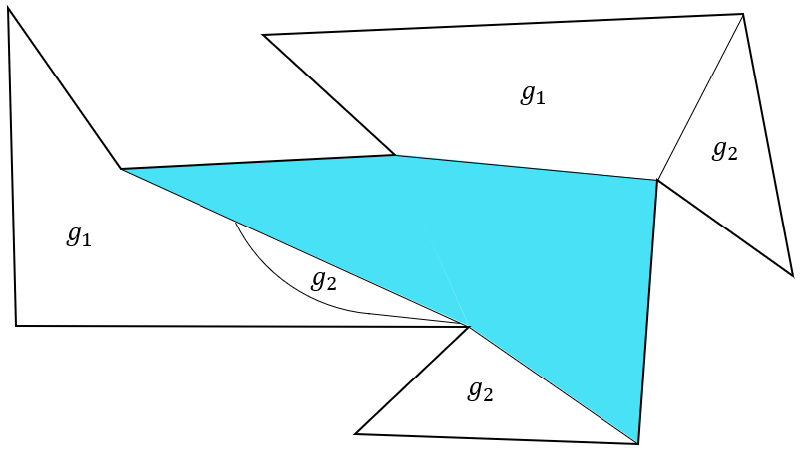}
			\caption{}
			\label{fig:ex_allocation2b}
		\end{subfigure}
		\begin{subfigure}{0.35\linewidth}
			\includegraphics[width=1\linewidth,height=1\linewidth]{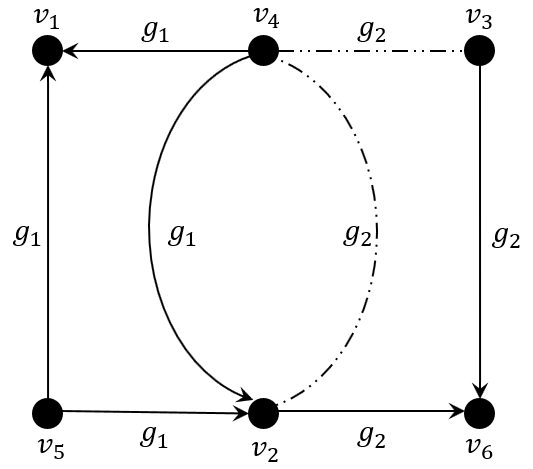}
			\caption{}
			\label{fig:ex_allocation2c}
		\end{subfigure}
	\end{center}
	\caption{(\subref{fig:ex_allocation2a}) Triangulated $P$ with guards deployed. (\subref{fig:ex_allocation2b}) The regions of the resulting partition after arbitrarily assigning $g_1$ to $T_4$. (\subref{fig:ex_allocation2c}) Resulting graph $G^{\#}$ after deleting a pair of edges and defining the orientation of all remaining edges.}
	\vspace{-0.1in}
	\label{fig:ex_allocation2}
\end{figure}

Figure \ref{fig:ex_allocation2} shows such a scenario involving $2$ guards and $6$ non-safe triangles. If Algorithm \ref{alg:orient} was executed, it could find $R_3^1(2)=T_3$ or $R_6^1(2)=T_6$, or also $R_1^1(1)=T_1$ or $R_5^1(1)=T_5$. The presence of the regular triangles $T_2$ and $T_4$ prevents Algorithm \ref{alg:orient} to find $\hat{U}_R^1(1)$ and $\hat{U}_R^1(2)$. There is no region that has initially been assigned to any of the guards inside $T_2$ and $T_4$ which can be used to construct the region of the other guard. Therefore, $\mathbb{G}_{ready}= \emptyset$ and $\mathbb{V}_{\Omega}=\{v_2,v_4\}$. Let $G_{\Omega}$ be the subgraph induced by $\mathbb{V}_{\Omega}$. We know that there are at least two guards with non-oriented edges incident to each $v_j \in \mathbb{V}_{\Omega}$. Since the number of vertices of $G_1$ is finite ($G_{\Omega}$ cannot be a tree with an infinite length path), each $v_j \in \mathbb{V}_{\Omega}$ is inside a cycle involving the diagonal of guards in $\mathbb{G}_{\Omega}$. Otherwise, if there is a $v_j \in \mathbb{V}_{\Omega}$ with only one neighbor in $\mathbb{V}_{\Omega}$, there is only one guard with a non-oriented edge incident to it which implies $v_j \notin \mathbb{V}_c$. In Figure \ref{fig:ex_allocation2} (c), $G_{\Omega}$ is the graph induced by $v_2$ and $v_4$. It is a cycle of length $2$ with edges $e_{2,4}(g_1)$ and $e_{2,4}(g_2)$.

Algorithm \ref{alg:arbitrary} resolves the aforementioned deadlock in $G_{\Omega}$. It arbitrarily selects a $g_i \in \mathbb{G}_{\Omega}$ and an endpoint of $h_i$. Subsequently, it assigns to $g_i$ all the unassigned regions that it can cover from the selected endpoint.  Thus, the edges of the other guards in $\mathbb{G}_{\Omega}$ that can cover those regions are deleted from $\mathbb{E}(G^{\#})$. Consequently, $g_i$ meets the requirement to get added to $\mathbb{G}_{ready}$ which in turn allows Algorithm \ref{alg:allocation} to continue. In the example of Figure \ref{fig:ex_allocation2}, Algorithm \ref{alg:arbitrary} arbitrarily selects $g_1$ and the endpoint that is a vertex of $T_4$ is chosen as $v_1(1)$. Thus, the edges $e_{4,2}(2)$ and $e_{4,3}(2)$ are deleted from $G^{\#}$, and $\mathbb{G}_{ready}=\{g_1\}$, so Algorithm \ref{alg:allocation} can continue. Next $g_1$ is chosen, and Algorithm \ref{alg:orient} finds $R_5^1(1)=T_5$, $R_4^1(1)=T_4$, $S_R^1(1)= \{ T_5,T_4 \}$ and $\hat{U}_R^1(1)=T_5 \cup T_4$. Thus, edges $e_{5,1}(1)$ and $e_{5,2}(1)$ become outgoing edges of $v_5$, and $e_{4,1}(1)$ and $e_{4,2}(1)$ become outgoing edges of $v_4$. It follows that $R_1^2(1)=T_1$ and $R_2^2(1)$ is the region inside $T_2$ assigned to $g_1$ shown in Figure \ref{fig:ex_allocation2} (b). At the end of this iteration, $\mathbb{G}_{alloc}=\{g_1\}$ and $\mathbb{G}_{ready}=\{g_2\}$. At the end of the second iteration all the regions of the environment have been assigned as shown in Figure \ref{fig:ex_allocation2} (b). The resulting directed graph $G^{\#}$ is illustrated in Figure \ref{fig:ex_allocation2} (c) in which the edges deleted by Algorithm \ref{alg:arbitrary} are shown as dotted segments.


Since the number of edges in $G^{\#}$ that can be associated to a given guard is upper bounded by $n(n-1)/2$, the time complexity of Algorithm \ref{alg:orient} is $O(n^2)$. Also, since the total number of edges that can be deleted is upper bounded by $O(n^2)$, the time complexity of Algorithm \ref{alg:arbitrary} is $O(n^2)$. At each iteration of Algorithm \ref{alg:allocation}, one guard enters $\mathbb{G}_{alloc}$ and Algorithm \ref{alg:orient} is called. It follows that the time complexity of Algorithm \ref{alg:allocation} is $O(n^3)$.


\begin{algorithm}
	\caption{ARBITALLOC}
	\begin{algorithmic}[1]
		\Procedure{1 of General Allocation}{$\mathbb{G}_{\Omega}$, $G^{\#}$}
		\State choose $g_i \in \mathbb{G}_{\Omega} $
		\State $v_1(i) \leftarrow$ arbitrary endpoint of $h_i$
		\For{each $T_k \in \overline{\mathbb{T}_{1}^{safe}}(g_i)$}
		\State $\mathbb{G}^{del}(T_k) \leftarrow (\mathbb{G}(T_k) \cap \mathbb{G}_{\Omega}) \backslash \{ g_i \})$
		\For{each $g_j \in \mathbb{G}^{del}(T_k)$}
		\State delete from $G^{\#}$ all edges incident to $v_k$ that correspond to $g_j$
		\EndFor
		\EndFor
		\EndProcedure
		
	\end{algorithmic}
	\label{alg:arbitrary}
\end{algorithm}




\subsection{Completeness and Correctness of GENALLOC}
\label{subsec:comp}


The following lemma proves that Algorithm \ref{alg:allocation} terminates. 

\begin{lemma}
	\label{lemma:5}
	Algorithm \ref{alg:allocation} terminates in a finite number of steps.
\end{lemma}
\begin{proof}
	Each iteration of Algorithm \ref{alg:allocation} computes the region to be allocated to an unassigned guard. Since the number of guards is finite, Algorithm \ref{alg:allocation} terminates in finite number of steps.\end{proof}
The next lemma states a condition under which Algorithm \ref{alg:allocation} is complete.

\begin{lemma}
	\label{lemma:6}
	If $G^{\#}$ is a forest, Algorithm \ref{alg:allocation} is complete.
\end{lemma}
\begin{proof} The proof is by contradiction. Consider a polygonal environment $P$ for which $G^{\#}$ is a forest, and Algorithm \ref{alg:allocation} fails to find a feasible allocation of guards even though one exists. From hereon, any variable with a symbol $\tilde{}$ on top of it is associated with the feasible allocation. For example, $\tilde{R}_j^2(i)$ denotes a region in $T_j$ assigned to $g_i \in \mathbb{G}(T_j)$ based on the feasible allocation. Since Algorithm \ref{alg:arbitrary} cannot find a feasible allocation in $P$, it terminates at Line $15$. Hence, there is a vertex $v_j \in \mathbb{V}(G^{\#})$ such that all the edges incident to $v_j$ are incoming edges, and the triangle $T_j \in \overline{\mathbb{T}^{safe}}$ has a region $R_{\emptyset}(j) \neq \emptyset$ which by definition cannot be assigned to any guard by Algorithm \ref{alg:arbitrary}. Since a feasible allocation exists, every point inside $R_{\emptyset}(j)$ can be assigned to a guard in $\mathbb{G}(T_j)$. Let $R_{\emptyset}^i(j)\subseteq R_{\emptyset}(j)$ denote the region that can be covered by a guard $g_i\in\mathbb{G}(T_j)$ in the feasible allocation. Therefore, $R_{\emptyset}^i(j)$ exists such that $R_{\emptyset}^i(j) \subseteq \tilde{R}_j^2(i)$ and $R_{\emptyset}^i(j) \not\subseteq {R}_j^2(i)$. For the feasible allocation, $d(\tilde{\hat{U}}_R^1(i),\tilde{R}_j^2(i)) \geq d_I^i$. Since $R_{\emptyset}(j) \cap \tilde{R}_j^2(i) \neq \emptyset$, $d(\hat{U}_R^1(i),\tilde{R}_j^2(i)) < d_I^i$. By definition, $d(\tilde{\hat{U}}_R^1(i),\tilde{R}_j^2(i)) \geq d_I^i$. Therefore, a region $R_{\emptyset}(k)\subset T_k \in \overline{\mathbb{T}_{1}^{safe}}(g_i)$ exists such that $R_{\emptyset}(k)\subset R^1_j(i)$ and $R_{\emptyset}(k)\not\subset \tilde{R}^1_j(i)$. Therefore, $R_{\emptyset}(k)\subset \bigcup_{g_a \in \mathbb{G}(T_k) \backslash \{ g_i \}}\tilde{R}^2_k(a)$. Hence, there exists a guard $g_a \in \mathbb{G}(T_k) \backslash \{ g_i \}$ and a region $R_{\emptyset}^a(k) \subseteq R_{\emptyset}(k)$ such that $R_{\emptyset}^a(k) \subseteq \tilde{R}_k^2(a)$ and $R_{\emptyset}^a(k) \not\subseteq {R}_k^2(a)$. Therefore, we can find a sequence $v_{i_1}\xleftarrow{g_{j_1}}v_{i_2}\xleftarrow{g_{j_2}}\cdots$ of vertices and guards in $G^{\#}$ such that $g_{j_k} \in \mathbb{G}(T_{i_k}) \backslash \{ g_{j_{k-1}} \}$, and corresponding regions $R_{\emptyset}^{j_k}(i_k) \subseteq R_{\emptyset}(i_k)$ such that $R_{\emptyset}^{j_k}(i_k) \subseteq \tilde{R}_{i_k}^2(j_k)$ and $R_{\emptyset}^{j_k}(i_k) \not\subseteq {R}_{i_k}^2(j_k)$. The sequence terminates if $\{g_{j_{k-1}}\} = \mathbb{G}(T_{i_k})$ in which case $T_{i_k}$ is an unsafe triangle. Since $j_k$ does not exist $\tilde{R}_{i_k}^2(j_k)=\emptyset\Rightarrow R_{\emptyset}^{j_k}(i_k)=\emptyset$. Since $R_{\emptyset}(i_k)=R_{\emptyset}^{j_k}(i_k)$ and $R_{\emptyset}(i_k)\neq \emptyset$, we arrive at a contradiction. If the sequence does not terminate, then there exist $i_k$ and $i_j$ such that $i_k=i_j$ since the number of vertices in $G^{\#}$ are finite. This implies the existence of a cycle in $G^{\#}$. Therefore, we arrive at a contradiction since $G^{\#}$ is a forest. 
\end{proof}

In the appendix (Lemma \ref{lemma:7}), we prove a more general result which shows that Algorithm \ref{alg:allocation} is complete if Algorithm \ref{alg:arbitrary} is never called during execution. 
Figure \ref{fig:simple_env} shows an example in which $G^{\#}$ does not contain cycles. In Figure \ref{fig:simple_env} (a), a simple polygonal environment is shown along with the deployment of guards. There are four non-safe triangles $T_1$, $T_2$, $T_3$ and $T_4$. $T_1$ and $T_2$ can be covered by $g_1$. $T_3$ and $T_4$ can be covered by $g_2$. In Figure \ref{fig:simple_env} (b), the corresponding graph $G^{\#}$ is shown. $G^{\#}$ is a forest that consists of two paths. Lemma \ref{lemma:6} states that Algorithm \ref{alg:allocation} will always find the feasible allocation for this environment.

\begin{figure}[thpb]
	\begin{center}
		\begin{subfigure}{0.6\linewidth}
			\includegraphics[width=0.92\linewidth,height=0.8\linewidth]{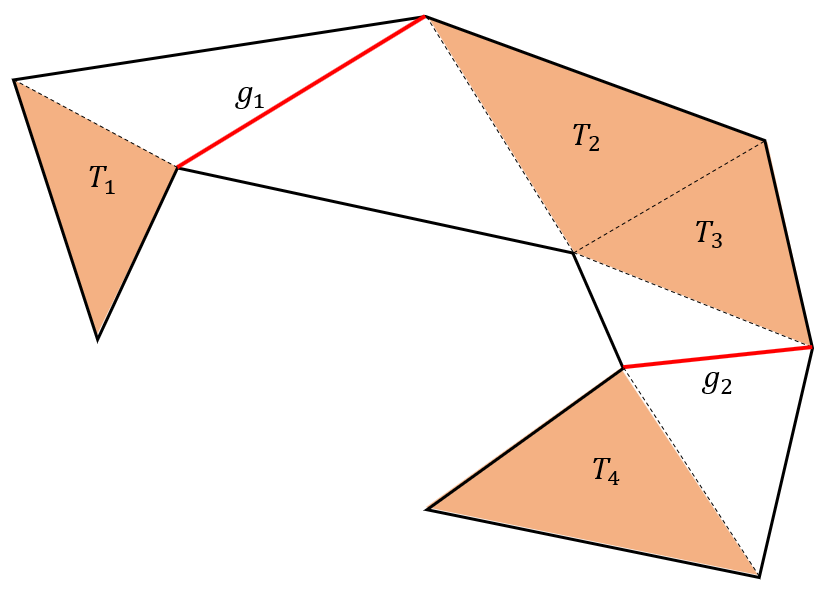}
			\caption{}
			\label{fig:simple_enva}
		\end{subfigure}
		\begin{subfigure}{0.34\linewidth}
			\includegraphics[width=0.9\linewidth,height=0.69\linewidth]{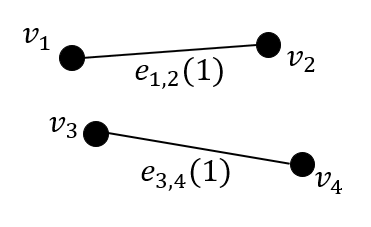}
			\caption{}
			\label{fig:simple_envb}
		\end{subfigure}
	\end{center}
	\caption{(\subref{fig:simple_enva}) A polygon is shown, the red segments represents the diagonals of the guards $g_1$ and $g_2$. (\subref{fig:simple_envb}) Graph $G_1$ that corresponds to the triangles and the deployment of the guards in the polygon.}
	\vspace{-0.1in}
	\label{fig:simple_env}
\end{figure}

\subsection{Classification of Guards}
\label{subsec:class}


In this section, we present a classification of the guards. It is based on the regions $\hat{U}_R^{\alpha}(i)$ $\alpha=\{1,2\}$ constructed from Algorithm \ref{alg:allocation}. 
\begin{enumerate}
	\item Type $0$ guard: These are guards for which either $\hat{U}_R^1(i) = \emptyset$ or $\hat{U}_R^2(i) = \emptyset$. Since the region allocated to a type $0$ guard can be covered from one endpoint of its diagonal, it is a static guard.
	
	\item Type $1$ guard: There are guards for which  all the non-safe triangles allocated to the guard incident to one endpoint are unsafe triangles. Notice that each edge in $\mathbb{E}(G^{\#})$ that corresponds to a type $1$ guard is an outgoing edge of a vertex $v_j$ such that $|\mathbb{G}(T_j)|=1$. Clearly, if $\mathbb{G}(T_j) \backslash \{ g_i \} = \emptyset$ for each $T_j \in \overline{\mathbb{T}_{1}^{safe}}(g_i)$, then $g_i$ is a type $1$ guard.
	
	\item Type $2$ guard: Any guard which is neither Type $0$ nor Type $1$ is a Type 2 guard. 
\end{enumerate}
Consider the example shown in Figure \ref{fig:type_2_guard}. Assume $T_1$, $T_2$ and $T_3$ are the only non-safe triangles incident to $v_1(1)$, $v_1(2)$ and $v_1(3)$ respectively. Algorithm \ref{alg:allocation}  selects the non-safe triangle $T_2$ to be assigned to the only guard that can cover it ($g_2$). Since $|\mathbb{G}(T_2)|=1$, $\hat{U}_R^{1}(2)=T_2$ $\Rightarrow$ $g_2$ is a type $1$ guard. Once $\hat{U}_R^{1}(1)\in T_1$ is computed after a few steps, Algorithm \ref{alg:allocation} selects the non-safe triangle $T_3$ and allocates the unshaded region to $g_3$. Besides $g_3$, $T_3$ can also be covered by $g_1$ and $g_2$. Regions $\hat{U}_R^{1}(1)$ and $\hat{U}_R^{1}(2)$ are known. Consequently, $R_{3}^{2}(1)$ and $R_{3}^{2}(2)$ are also computed. It follows that the unshaded region in $T_3$ is labeled as $R_{3}^{1}(3)$. Moreover, since $\overline{\mathbb{T}_{1}^{safe}}(g_3)=\{T_3\}$, $\hat{U}_R^{1}(3)=R_{3}^{1}(3)$. It follows that $g_3$ is an example of a type $2$ guard.

\begin{figure}[htb]
	\begin{center} 
		\includegraphics[width=0.85\linewidth,height=0.45\linewidth]{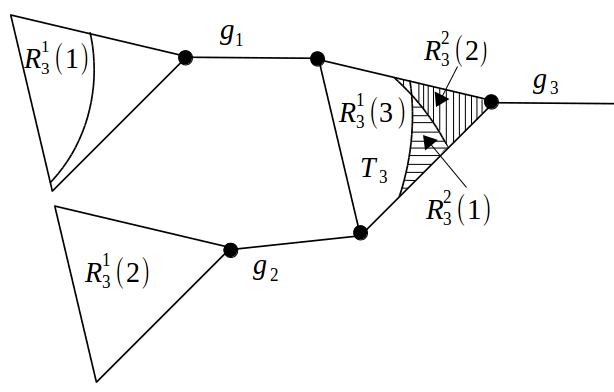}
	\end{center}
	\caption{Example of a type $1$ guard $g_2$ and a type $2$ guard $g_3$.}
	\label{fig:type_2_guard}
\end{figure}

In Algorithm \ref{alg:arbitrary} an unassigned guard $g_i \in \mathbb{G}$ is arbitrarily chosen to cover the unassigned regions inside the triangles incident to one of the endpoints of its diagonal. $g_i$ is arbitrarily chosen since Algorithm \ref{alg:allocation} failed to find a unique partition of such regions. It implies that $\hat{U}_R^{1}(i)$ cannot be constructed unlike in the case of a type $0$, a type $1$ nor a type $2$ guard. However, the arbitrary allocation in Algorithm \ref{alg:arbitrary} assigns those regions to $g_i$. As a result, $g_i$ is converted to a type $1$ or a type $2$ guard.

\subsection{Motion strategy for the Guards}
\label{subsec:react}

In this section, we present a motion strategy for the guards to move on their diagonals. We introduce the concept of {\it critical curves} to propose activation strategies for type $1$ and type $2$ guards. 

From the discussion in the previous section, $\hat{U}_R^1(i)$ is the region assigned to a type $1$ or type $2$ guard $g_i$. Therefore, it is the responsibility of $g_i$ to cover the triangles incident to $v_1(i)$ when the intruder lies in $\hat{U}_R^1(i)$. 

We define an {\it internal critical curve}, denoted by $s_{int}^1(i)$, as the boundary of $\hat{U}_R^1(i)$. Corresponding to an internal critical curve, we define an {\it external critical curve} as follows:

\begin{equation}
\label{eq:ext}
s_{ext}^1(i)=\{p \in P \backslash \hat{U}_R^1(i) : d(p,s_{int}^1(i)) =  d_I^i \}\\
\end{equation}
Comparing the definition of $s_{ext}^1(i)$ to the definition of $R_j^2(i)$ (see Equation (\ref{eq:r_2})), we can conclude that a part of the boundary of $R_j^2$ can belong to $s_{ext}^1(i)$. We define a \textit{critical region} associated with the guard $g_i$ as follows: 
\begin{equation}
\label{eq:critical}
C_1(i)=\{p \in P \backslash \mathring{\hat{U}}_R^1(i) : d(p,s_{int}^1(i)) \leq  d_I^i \},
\end{equation}
where $\mathring{\hat{U}}_R^1(i)$ is the interior of $\hat{U}_R^1(i)$. Note that, by definition, the boundary of $C_1(i)$ contains both curves $s_{int}^1(i)$ and $s_{ext}^1(i)$, and since $d(\hat{U}_R^1(i),\hat{U}_R^2(i)) \geq d_I^i$, it is clear that $(\hat{U}_R^2(i) \cap C_1(i)) \subset s_{ext}^1(i)$. 

Figure \ref{fig:critical} shows the region $\hat{U}_R^1(i)=R_1^1(i) \cup R_2^1(i) \cup R_3^1(i)$ for the guard $g_i$. The neighboring guards are not shown for the sake of simplicity. $s_{int}^1(i)$, the boundary of $\hat{U}_R^1(i)$, is represented as blue segments and arcs that form the boundary of the regions $R_1^1(i)$, $R_2^1(i)$ and $R_3^1(i)$. Since triangles $T_4$ and $T_5$ are safe triangles (they have $h_i$ as an edge), there is no internal critical curve inside them. The green segments and curves denote $s_{ext}^1(i)$, and the unshaded region inside $P$ is $C_1(i)$. The boundary of $C_1(i)$ is formed by $s_{int}^1(i)$, $s_{ext}^1(i)$ and edges of the environment. The dark colored regions represent $R_6^2(i)$ and $R_7^2(i)$ which are part of $\hat{U}_R^2(i)$.

\begin{figure}[htb]
	\begin{center} 
		\includegraphics[width=0.74\linewidth,height=0.5\linewidth]{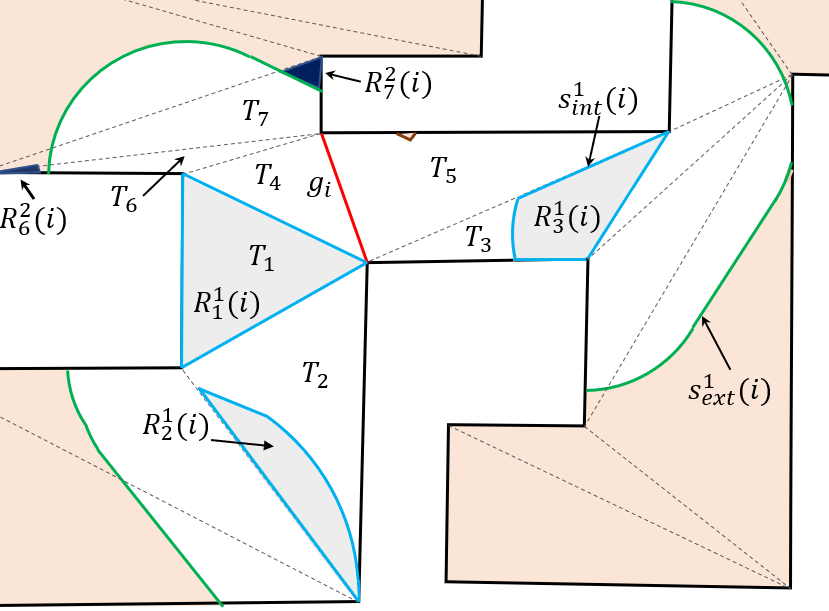}
	\end{center}
	\caption{Example of the definition of the critical region of guard.}
	\label{fig:critical}
\end{figure}


For an intruder located in $C_1(i)$, the following equation maps the position of the intruder ($x_I$) to the position of $g_i$ (denoted as $x_{g_i}$) along its diagonal:	

\begin{equation}
\label{eq:1}
x_{g_i}= x_{v_1(i)}+ \frac{d(s_{int}^{1}(i),x_I)}{d_I^i}(x_{v_2(i)}-x_{v_1(i)}),
\end{equation}
where $x_{v_\alpha(i)}$ is the location of vertex $v_{\alpha}(i)$ ($\alpha =\{1,2\}$). If $x_I \in \hat{U}_R^1(i)$, $g_i$ remains static at $v_1(i)$. Otherwise, if $x_I(t) \notin (\hat{U}_R^1(i) \cup C_j(i))$, $g_i$ remains static at $v_2(i)$.

By definition, $d(\hat{U}_R^1(i),\hat{U}_R^2(i)) \geq d_I^i$ while $d(\hat{U}_R^1(i),s_{ext}^1(i))=d_I^i$. Hence, (\ref{eq:1}) guarantees that $g_i$ will always cover the regions assigned to it when the intruder is located in them. 


Consider the case of a guard $g_i$ such that $g_i$ is incident to $T_k$ at a vertex $v_2(i)$. The motion strategy proposed in (\ref{eq:1}) ensures that $\hat{U}_R^1(i)$ is covered if there is an intruder inside $\hat{U}_R^1(i)$. However, if the intruder is located in $\hat{U}_R^1(i) \cup \mathring{C}_1(i)$, where $\mathring{C}_1(i)$ is the interior of $C_1(i)$, $g_i$ cannot cover $T_k$ because $g_i$ can only be located at $v_2(i)$ when the intruder is outside $\hat{U}_R^1(i) \cup \mathring{C}_1(i)$ according to (\ref{eq:1}). Now, consider the case when $g_i$ is incident to $T_k$ at a vertex labeled $v_1(i)$. Since $v_1(i)$ is a vertex of $T_k$, $T_k \cap \hat{U}_R^1(i)$ will be covered by $g_i$ if there is an intruder inside it according to (\ref{eq:1}). However, if the intruder lies outside $T_k \cap \hat{U}_R^1(i)$, then $T_k$ is not covered by $g_i$ since $g_i$ is not located at $v_1(i)$ according to (\ref{eq:1}). Consequently, for each $g_i\in\mathbb{G}(T_k)$, there exists a region which prevents $g_i$ to cover $T_k$ when the intruder lies inside it. This region, denoted by $\hat{C}_1(i)$, is called the \textit{extended critical region}. It is given by the following expression:

\begin{equation}
\hat{C}_1(i)=
\left\{
\begin{array}{cc}
\hat{U}_R^1(i) \cup \mathring{C}_1(i) & v_2(i) \text{ is a vertex of }  T_k\\
P \backslash \hat{U}_R^1(i) & \text{otherwise}
\end{array}
\right.,
\end{equation}

Based on the concept of extended critical regions, Lemma \ref{lemma:5.10} presents a necessary and sufficient condition for the guards to cover a non-safe triangle when an intruder lies in it.

\begin{lemma} 
	\label{lemma:5.10}
	For the guards in $\mathbb{G}(T_k)$ ($T_k \in \overline{\mathbb{T}^{safe}}(G)$), (\ref{eq:1}) guarantees that the triangle $T_k$ is covered when an intruder is located in it if and only if $\displaystyle\bigcap_{g_i \in \mathbb{G}(T)}( \hat{C}_1(i) \cap T_k) = \emptyset$.		
\end{lemma}
\begin{proof}
	($\Rightarrow$) Assume that $\bigcap_{g_i \in \mathbb{G}(T_k)} (\hat{C}_1(i) \cap T_k)= \emptyset$, and $T_k$ is not covered when the intruder lies in $T_k$. 
	It implies that there is a location inside $T_k$ for the intruder that prevents every $g_i\in \mathbb{G}(T_k)$ to cover $T_k$. According to (\ref{eq:1}), such a region must belong to $\bigcap_{g_i \in \mathbb{G}(T_k)} (\hat{C}_1(i) \cap T_k)$ which contradicts our assumption. 
	
	($\Leftarrow$) Next, assume that $\bigcap_{g_i \in \mathbb{G}(T_k)} (\hat{C}_1(i) \cap T_k) \neq \emptyset$ and $T_k$ is covered when the intruder is located in it. Since $\bigcap_{g_i \in \mathbb{G}(T_k)} \hat{C}_1(i) \cap T_k \neq \emptyset$ when $x_I \in \bigcap_{g_i \in \mathbb{G}(T_k)} \hat{C}_1(i) \cap T_k$, there is no guard covering $T_k$ according to (\ref{eq:1}), which is a contradiction. The lemma follows.
\end{proof}

\section{Polygons with Holes}
\label{sec:holes}

In this section, we assume that $P$ has polygonal holes which represent obstacles inside the polygon. Let $\mathbb{Q}=\{Q_1,\ldots,Q_N\}$ represent the set of polygonal holes. Let $\hat{n}=n+n_{Q_1}+\ldots+n_{Q_N}$ denote the total number of vertices of $G$, where $n_i$ is the number of vertices of the hole $Q_i$ and $i \in \{ 1,2,\ldots,N \}$. Figure \ref{fig:simply} shows a polygonal environment with an internal polygonal hole $Q_1$. In Theorem $5.1$ of \cite{O'Rourke:1987}, it is shown that one can find an internal diagonal of the triangulation of $P$ between any two holes (or between a hole and the outer boundary) which merges two holes (or the hole with the boundary) if a wall of thickness $0$ is placed on the diagonal. This reduces the value of $N$ by $1$. See Figure \ref{fig:simply} where $Q_1$ is merged with the outer boundary through the diagonal shared by triangles $T_1$ and $T_2$. Therefore, for any polygon $P$ with $n$ vertices and $N$ internal polygonal holes, we can construct a simply-connected polygon $P'$ with $n+2N$ vertices. We can apply all techniques proposed in the previous sections for deploying guards, and allocating them to triangles of the triangulation of $P'$ for tracking the intruder.


\begin{figure}[thpb]
	\begin{center}
		\includegraphics[width=0.80\linewidth,height=0.54\linewidth]{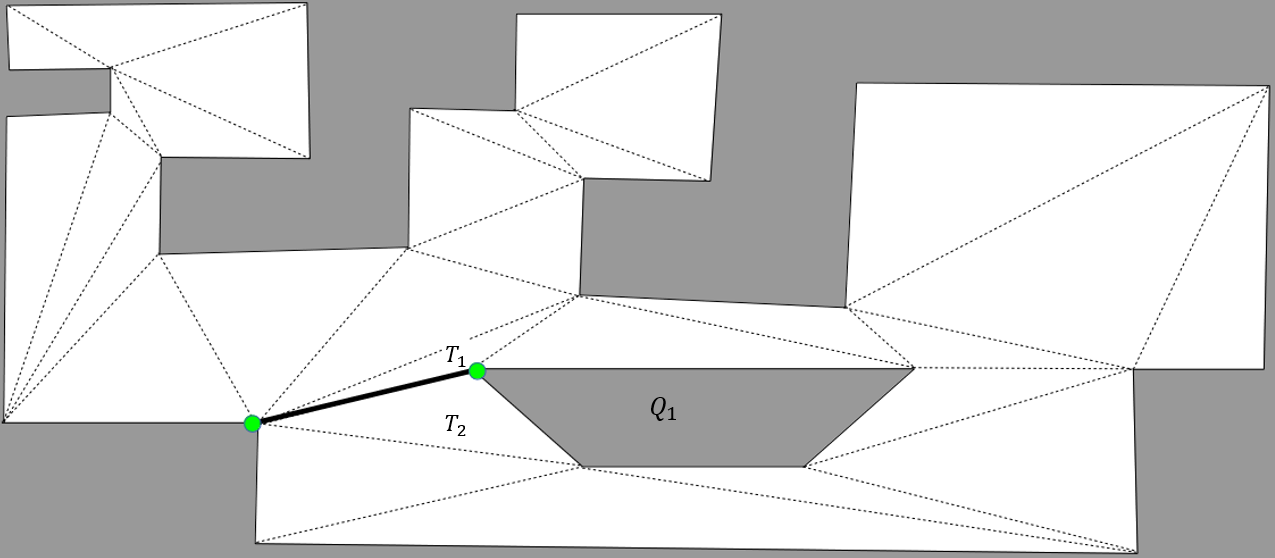}
	\end{center}
	\caption{(a) There is always a diagonal connecting $Q_1$ with the boundary of $P$.}
	\vspace{-0.1in}
	\label{fig:simply}
\end{figure}

\section{Tracking Multiple intruders}
\label{sec:multi}

In this section, we analyze the performance of the proposed algorithm for multiple intruders. We assume that all the intruders have the same maximum speed $\bar{v}_e$. We use the symbol $\mathbb{I}$ $(|\mathbb{I}| > 1)$ to denote the set of intruders, and the vector $x_{\mathbb{I}}(t)\in\mathbb{R}^{|\mathbb{I}|}$ to denote their positions inside the polygon. We assume that the deployment of the guards and the allocation of the different regions of the environment are obtained using the techniques presented in sections \ref{subsec:dep} and \ref{sec:alloc}, respectively.

In Section \ref{subsec:react}, a motion strategy for the guards was proposed for a single intruder. (\ref{eq:1}) is a reactive motion strategy that depends on the location of the intruder. It ensures that each guard $g_i$ can cover $\hat{U}_R^1(i)$ when the intruder is inside it. In the presence of multiple intruders, the priority for each $g_i$ is to cover $\hat{U}_R^1(i)$ as long as there is an intruder inside it. Hence, the motion strategy only needs to consider the intruder closer to $\hat{U}_R^1(i)$, i.e. if there is an intruder inside $\hat{U}_R^1(i)$, $g_i$ stays at $v_1(i)$ regardless of the positions of other intruders. Therefore, the motion strategy of the guard in this case is given by (\ref{eq:3}), wherein $d(s_{int}^{1}(i),x_I)$ is replaced by $d_{min}(s_{int}^{1}(i),x_{\mathbb{I}}(t))$ defined as follows:

\begin{equation}
\label{eq:3}
d_{min}(i,x_{\mathbb{I}}(t))= \min_{I_k\in\mathbb{I} }{d(s_{int}^{1}(i),x_{I_k}(t))}.
\end{equation}

In Section \ref{subsec:react}, we showed that an intruder inside each $\hat{C}_1(i)$ associated to $g_i \in \mathbb{G}(T_k)$ will prevent anon-safe triangle $T_k$ from being covered by any guard incident to it. Therefore, $| \mathbb{G}(T_k)|$ intruders are sufficient to keep $T_k$ uncovered. However, if for example, there are two guards $g_i,g_j \in \mathbb{G}(T_k)$ such that $\hat{C}_1(i) \bigcap \hat{C}_1(j)\neq\emptyset$, a single intruder lying inside the intersection will prevent $g_i$ and $g_j$ from covering $T_k$. Therefore, fewer than $|\mathbb{G}(T_k)|$ intruders can prevent $T_k$ from being covered by any guard incident to it if there are non-empty intersections between the extended critical regions corresponding to distinct guards that can cover $T_k$.


Consider the power set $2^{\mathbb{G}(T_k)}$ of all guards incident to $T_k$. Let $\mathcal{S}\subseteq 2^{\mathbb{G}(T_k)}$ be a collection of all sets $S\in 2^{\mathbb{G}(T_k)}$ for which the extended critical regions of the guards belonging to $S$ have a non-empty intersection. The problem of finding the minimum number of intruders that can be placed at the intersection of extended critical regions to uncover $T_k$ is equivalent to the problem of finding the minimum cover $\mathcal{C}\subseteq \mathcal{S}$ of $\mathbb{G}(T_k)$.

Let $n_I(T_k)$ denote the maximum number of intruders that can be tracked by the guards incident to $T_k$ without uncovering $T_k$ when there is an intruder inside it. The following lemma relates $n_I(T_k)$ to $|\mathcal{C}|$.

\begin{lemma}
	\label{lemma:6.3}
	Let $T_k \in \overline{\mathbb{T}^{safe}}(G)$. 
	\begin{equation}n_I(T_k)=\left\{\begin{array}{cc}|\mathcal{C}|&T_k \cap \mathcal{I}(C) = \emptyset\text{, }\forall C\in \mathcal{C}\\
	|\mathcal{C}|-1& \text{otherwise,}
	\end{array}\right.\end{equation} 
	where $\mathcal{I}(C)=\displaystyle\bigcup_{g_k  \in C}{\hat{C}_1(k)}$
	
\end{lemma}
\begin{proof} $|\mathcal{C}|$ intruders, each placed in a distinct $\mathcal{I}(C)$, are sufficient to prevent all guards incident to $T_k$ from covering it when an intruder is located inside $T_k$. If $T_k\cap\mathcal{I}(C)=\emptyset\text{ }\forall C\in \mathcal{C}$, the intruder inside the $T_k$ cannot lie inside any 
	$\mathcal{I}(C), C\in\mathcal{C}$. Therefore, $n_I(T_k)=|\mathcal{C}|$ in this case. Otherwise, the intruder located inside $T_k$ can cover an $\mathcal{I}(C)$. Therefore, $n_I(T_k)<|\mathcal{C}|$ in this case.  
	
	Since $\mathcal{C}$ is the minimum set cover of $\mathbb{G}$, $|\mathcal{C}|-1$ intruders cannot prevent $T_k$ from being covered by at least one guard $g_{i}\in\mathbb{G}(T_k)$. Therefore, $n_I(T_k)\geq|\mathcal{C}|-1$. The theorem follows.
\end{proof}

In Figure \ref{fig:examplem} \subref{fig:examplema}, two type $1$ guards $g_1$ and $g_2$, and a type $2$ guard $g_3$ are shown with their corresponding external and internal critical curves. The corresponding endpoints $v_1(i)$ are shown with a green disc. The regions $\hat{U}_R^1(i)$ are shaded in orange, and the safe triangles are shaded in blue. $T$ is a regular triangle that can be covered by $g_1$, $g_2$ and $g_3$. Therefore, the set $\mathcal{C}$ consists of the external critical regions and has cardinality $3$. Since there is no intersection between any region in $\mathcal{C}$ and $T$ then $n_I(T)=3$. Figure \ref{fig:examplem} \subref{fig:examplemb}, shows the same case but for a smaller value of $r$. $s_{ext}^1(1) \cap s_{ext}^1(2) \neq \emptyset$. This implies that $\mathcal{C}$ consists of the extended critical region of $g_3$ and the intersection of the extended critical regions of $g_1$ and $g_2$. In this case, $\mathcal{C}$ consists of $\hat{C}_1(4)$ and $\hat{C}_1(1) \cap \hat{C}_1(2)$, so it has cardinality $2$. Since none of the regions in $\mathcal{C}$ intersects with $T$, then $n_I(T)=2$.


\begin{figure}[thpb]
	\begin{center}
		\begin{subfigure}{0.47\linewidth}
			\includegraphics[width=1\linewidth,height=1\linewidth]{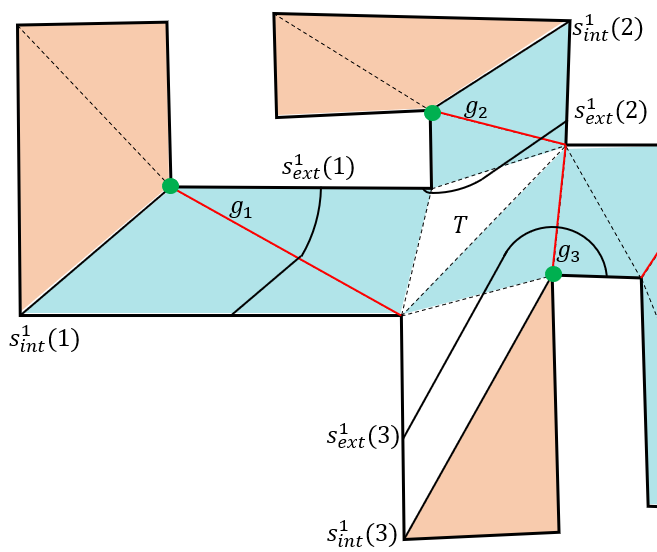}
			\caption{}
			\label{fig:examplema}
		\end{subfigure}
		\begin{subfigure}{0.45\linewidth}
			\includegraphics[width=1\linewidth,height=1\linewidth]{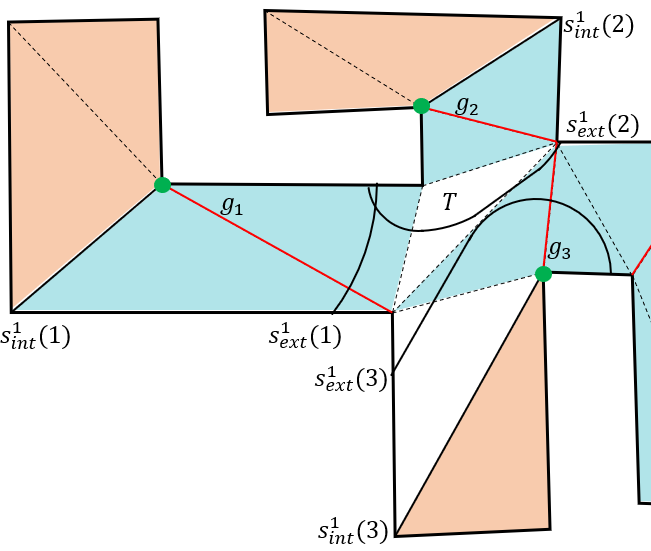}
			\caption{}
			\label{fig:examplemb}
		\end{subfigure}
	\end{center}
	\caption{(\subref{fig:examplema}) An instance where the regions $\hat{C}_1(i)$ do not intersect. (\subref{fig:examplemb}) An instance where two regions $\hat{C}_1(i)$ intersect.}
	\vspace{-0.1in}
	\label{fig:examplem}
\end{figure}

\begin{theorem} 
	\label{corollary:1}
	The minimum number of intruders that can be tracked based on the strategy proposed in (\ref{eq:3}) is $n_I^* = \min\{ n_I(T_k) : T_k \in \overline{\mathbb{T}^{safe}}(G) \}$.
\end{theorem}
\begin{proof} 
	Assume that $|\mathbb{I}| > \min\{ n_I(T_k) : T_k \in \overline{\mathbb{T}^{safe}}(G) \}$. It implies that there is at least one $T_k \in \overline{\mathbb{T}^{safe}}(G)$ for which $n_I(T_k) < |\mathbb{I}|$ and therefore, $T_k$ cannot be covered at all times according to Lemma \ref{lemma:6.3}$\implies$ $n_I^* \ngtr \min\{ n_I(T_k) : T_k \in \overline{\mathbb{T}^{safe}}(G) \}$. Now assume that $ |\mathbb{I}| \leq \min\{ n_I(T_k) : T_k \in \overline{\mathbb{T}^{safe}}(G) \}$. According to Lemma \ref{lemma:6.3}, the guards have a strategy to cover every non-safe triangle if $ |\mathbb{I}| \leq \min\{ n_I(T_k) : T_k \in \overline{\mathbb{T}^{safe}}(G) \}$. Therefore, $n_I^* = \min\{ n_I(T_k) : T_k \in \overline{\mathbb{T}^{safe}}(G) \}$.
\end{proof}

The set cover problem is NP-complete \cite{Karp:1972}. Several polynomial time approximation schemes (PTAS) for the set cover problem have been proposed in the literature \cite{Chvatal:1979,Gandhi:2006,Goldschmidt:1993,Johnson:1974,Lovasz:1975,Slavak:1997}.  Better approximation ratios can be obtained at the expense of computational complexity slightly higher than a PTAS \cite{Bougeois:2009}. For example, it has been shown that any $(1 - \alpha \ln {n})$ - approximation algorithm for the set cover problem must run in time at least $2^{n^{c \alpha}}$ for some small constants $0 < c < 1$ \cite{Moshkovitz:2012}. \cite{Cygan:2008,Cygan:2009} present some efforts to tighten the running time by reducing the value of $c$ in $2^{n^{c \alpha}}$. We can use either of the aforementioned approaches to obtain $\mathcal{C}$. For a problem instance of large size, one might prefer a PTAS, whereas a moderately exponential algorithm is more preferable when the number of guards covering each $T_k$ is small enough.


\section{Conclusion}
\label{sec:conclusion}

In this work, we addressed the problem of tracking mobile intruders in a polygonal environment using a team of diagonal guards. Leveraging on deployment strategies for mobile coverage in art gallery problems, we proposed control and coordination strategies for the guards to track intruders inside a polygonal environment. At first, we formulated the tracking problem as a multi-robot task assignment problem {\bf on the triangulation graph of a polygon}. We classified the guards based on their position with respect to the triangles of the triangulation. Next, we showed that the problem of finding the minimum speed of the guards to cover the triangles of the triangulation under the constraint that each triangle can only be covered by a single guard is NP-hard. Given the maximum speed of the intruder, we proposed an algorithm to find a feasible allocation of guards to the triangles of the triangulation when multiple guards are allowed to cover the triangle. We proved the correctness of the proposed algorithm, and its completeness for a specific set of inputs. Based on the task allocated to a guard, we proposed control laws for the guards to move along their diagonals. Finally, we extended the algorithm to address deployment and allocation strategies for non-simple polygons and multiple intruders. 



We believe that our paper is a first step towards MRS deployment for {\bf persistent tracking with provable guarantees}. An important direction of future research is to address the tracking problem for guards with sensing and motion constraints, for example, edge guards, which are more constrained in their motion, or line guards, which are less constrained than diagonal guards. Another future research direction is to study the tracking problem for special polygons, for example, orthogonal polygons, monotone polygons etc. For these polygons, it has been shown that fewer guards are required for coverage. Finally, the problem of tracking with mixed team of guards (static and mobile) is an interesting direction of future research.

\bibliographystyle{unsrt}
\bibliography{new_references}

\appendix

\section{Construction of $\hat{U}^{2}_R(i)$ and $R^{2}_j(i)$}
\label{sec:appa}


Obtaining all the regions $R_j^{2}(k)$ yields the set $S_R^{2}(i)$ and the region $\hat{U}_R^2(i)$. 

We claim that in general, the boundary of any region $\hat{U}_R^{1}(i)$, denoted by $\delta (\hat{U}_R^{1}(i))$, consists of arcs of circle. Let $\mathbb{S}(\delta(\hat{U}_R^1(i)))$ be the set of arcs of circle of $\delta(\hat{U}_R^1(i))$. For each $s \in \mathbb{S}(\delta(\hat{U}_R^1(i)))$, we also define $c(s)$ as the center of the circle generating $s$ (center at infinite in the case of a line segment), $rad(s)$ as its radius. We define the \textit{expanded boundary} of $\hat{U}_R^1(i)$ as $\gamma (\hat{U}_R^1(i))=\{ p \in P \backslash \hat{U}_R^1(i): d(p,\delta (\hat{U}_R^1(i)))=d_I^i \}$. Lemma \ref{lemma:boundary} shows that $\gamma(\hat{U}_R^1(i))=s_{int}^1(i)$ consists only of arcs of circle grouped in a set $\mathbb{S} (\gamma(\hat{U}_R^1(i)))$, and Lemma \ref{lemma:boundary2} shows that the boundary of $\hat{U}_R^{1}(i)$ and $\hat{U}_R^{2}(i)$ consist of arcs of circle.

It is important to make this remark since most of the computational geometry libraries include segments and circle arcs as basic classes, which are required to build regions $\hat{U}_R^{1}(i)$ and $\hat{U}_R^{2}(i)$. Some computational geometry libraries such as CGAL \cite{Wein:2007} and \cite{Mehlhorn:1999}, include the implementation of approximation techniques to compute offset curves of polygons. For the case of offsets of polylines there are some approximation algorithms \cite{Liu:2007,Choi:1999,Jian:2001} which may be implemented using the aforementioned libraries with their line segment and circle classes. For this paper we used the LEDA $6.5$ library in the simulations.

\begin{lemma}
	\label{lemma:boundary}
	Given $\mathbb{S}(\delta(\hat{U}_R^1(i)))$, its corresponding $\gamma(\hat{U}_R^1(i))$ consists of arcs of circle. 
\end{lemma}

\begin{proof}
	Trivially, if there are no obstacles between $\delta(\hat{U}_R^1(i))$ and $\gamma(\hat{U}_R^1(i))$, $\gamma (\hat{U}_R^1(i))=\{ p \in P \backslash \hat{U}_R^1(i): d(p,\delta (\hat{U}_R^1(i)))=d_I^i \}$ is the offset of $\delta(\hat{U}_R^1(i))$ (with $d_I^i$ as the offset distance). Therefore, $\gamma (\hat{U}_R^1(i))$ must be a polyline curve \cite{Liu:2007}.
	
	The presence of obstacles between $\delta(\hat{U}_R^1(i))$ and $\gamma(\hat{U}_R^1(i))$ implies that the shortest path from some points in $\gamma(\hat{U}_R^1(i))$ to $\delta(\hat{U}_R^1(i))$ is a chain of connected line segments instead of a line segment, as in the case where there are no obstacles between $\gamma(\hat{U}_R^1(i))$ and $\delta(\hat{U}_R^1(i))$. In Figure \ref{fig:bound} a region $\hat{U}_R^1(i)$ is shown as an orange triangle, its corresponding $\delta(\hat{U}_R^1(i))$ is a black dotted segment, $\gamma(\hat{U}_R^1(i))$ is represented as a dotted curve divided into four arcs of circle. $s_1 \in \mathbb{S}(\gamma(\hat{U}_R^1(i)))$ illustrates the case of points in $\gamma(\hat{U}_R^1(i))$ such that the shortest path between them and $\delta(\hat{U}_R^1(i))$ is a line segment. Now consider $s_2 \in \mathbb{S}(\gamma(\hat{U}_R^1(i)))$ the presence of an obstacle implies that for all the points in $s_2$, such as the one illustrated as a black circle, the shortest path between $\gamma(\hat{U}_R^1(i))$ and $\delta(\hat{U}_R^1(i))$ consists of two connected line segments, one with endpoints in $s_2$ and vertex $v_1 \in \mathbb{V}(G)$ and other with $v_1$ as an endpoint and the other at $\delta(\hat{U}_R^1(i))$. Also, for the points in $s_3 \in \mathbb{S}(\gamma(\hat{U}_R^1(i)))$, the shortest path between $\gamma(\hat{U}_R^1(i))$ and $\delta(\hat{U}_R^1(i))$ consists of three connected line segments, one from $s_3$ to $v_2$, another from $v_2$ to $v_1$ and the last one from $v_1$ to $\delta(\hat{U}_R^1(i))$. Consider the points $p \in \gamma(\hat{U}_R^1(i))$ for which the shortest path between them and $\delta(\hat{U}_R^1(i))$ is not a line segment due to the presence of reflex vertices of the environment. As we can see in Figure \ref{fig:bound}, for any of such $p$ points, the shortest path from $p$ to $\delta(\hat{U}_R^1(i))$ must visit first a reflex vertex $v_{p} \in \mathbb{V}(G)$ ($v_1$ or $v_2$ for instance) of the environment. Clearly, the union of such points $p$ for which the first segment of the shortest path between them and $\delta(\hat{U}_R^1(i))$ has $v_{p}$ as an endpoint, is a subset of the union of all points in the plane that are equidistant to $v_{p}$. Thus, they form an arc of circle, which is centered at $v_{p}$ with radius $d_I^i- d(v_{p}, \delta(\hat{U}_R^1(i)))$, where $d(v_{p}, \delta(\hat{U}_R^1(i)))$ is the length of the shortest path between $v_{p}$ and $\delta(\hat{U}_R^1(i))$. Therefore, every point in $\gamma(\hat{U}_R^1(i))$ belongs to an arc of circle, so $\gamma(\hat{U}_R^1(i))$ is the union of a set $\mathbb{S}(\gamma(\hat{U}_R^1(i)))$ of arcs of circle.
\end{proof}

\begin{figure}[htb]
	\begin{center} 
		\includegraphics[width=0.62\linewidth,height=0.43\linewidth]{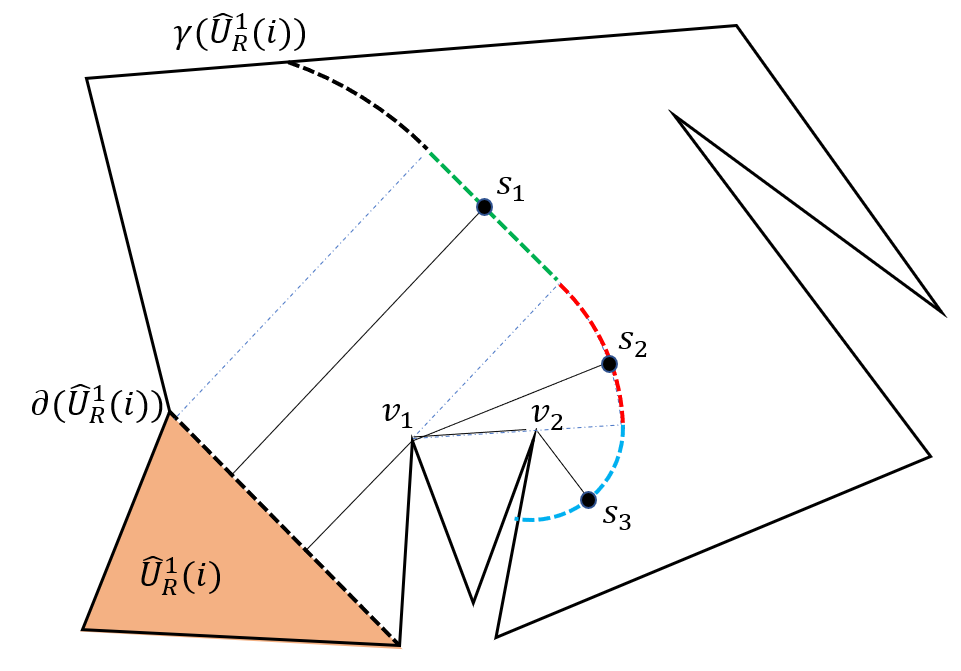}
	\end{center}
	\caption{Boundary $\gamma(\hat{U}_R^1(i))$ decomposed in four different arcs of circle.}
	\label{fig:bound}
\end{figure}

\begin{lemma}
	\label{lemma:boundary2}
	Every $\hat{U}_R^1(i)$ and $\hat{U}_R^2(i)$ is bounded by arcs of circle.
\end{lemma}
\begin{proof}	
	Assume that $\hat{U}_R^1(i)$ is bounded by arcs of circle. According to (\ref{eq:r_2}), $R_j^2(i)$ where $T_j \in \overline{\mathbb{T}_{2}^{safe}}(g_i)$, is the intersection of $T_j$ and the complement of the region enclosed by $\delta (\hat{U}_R^{1}(i))$ and $\gamma (\hat{U}_R^{1}(i))$. Hence the boundary of $R_j^2(i)$ consists of arcs of circle in $\mathbb{S}(\gamma(\hat{U}_R^1(i)))$ and the edges of $T_j$ (arcs of circle with center at infinity). Since the boundary of $R_j^2(i)$ consists of arcs of circle, it follows that $\delta (\hat{U}_R^2(i))$ also consists of arcs of circle from the definition of $\hat{U}_R^2(i)$. Consider the base case where $\hat{U}_R^{1}(i)$ is the union of unsafe triangles. $\delta (\hat{U}_R^{1}(i))$ is then a set of line segments, and the lemma holds. Consider the case of an unassigned guard $g_i$ such that all the guards $g_k \in \mathbb{G}(T_j)\backslash \{ g_i \}$ have their regions $\hat{U}_R^1(k)$ and $\hat{U}_R^2(k)$ defined for each $T_j \in \overline{\mathbb{T}_{\alpha}^{safe}}(g_i)$, where $\alpha = 1 \vee 2 $. We assume that the result holds for the regions of those guards $g_k \in \mathbb{G}(T_j)\backslash \{ g_i \}$. Since the regions of those guards are defined, then $\alpha=1$ and each $R_j^{1}(i)$ can be defined. Recall that $R_j^{1}(i) = \bigcap_{g_k \in \mathbb{G}(T_j) \backslash \{ g_i \} }{\overline{R}_j^2(k)} \cap T_j$, where $T_j \in \overline{\mathbb{T}_{1}^{safe}}(g_i)$. Each $R_j^{1}(i)$ is an intersection of regions $\hat{U}_R^{2}(k)$ which are bounded by arcs of circle. Hence, the intersection region is also bounded by arcs of circle, which trivially implies that $\delta (\hat{U}_R^{1}(i))$ also consists of arcs of circle. And according to the first part of this proof, it follows that $\delta (\hat{U}_R^2(i))$ consists of arcs of circle.
\end{proof}

\section{Generalization of Lemma \ref{lemma:6}}
\label{sec:appb}
\begin{lemma}
	\label{lemma:7}
	If Algorithm \ref{alg:arbitrary} is never called for a specific input, Algorithm \ref{alg:allocation} is complete for such an input.
\end{lemma}
\begin{proof} The proof is by contradiction. We assume that Algorithm \ref{alg:allocation} does not find a feasible allocation but there exists one. Additionally, we assume that Algorithm \ref{alg:arbitrary} is never called during execution. In the proof of Lemma \ref{lemma:6}, we show the existence of a sequence $v_{i_1}\xleftarrow{g_{j_1}}v_{i_2}\xleftarrow{g_{j_2}}\cdots$ of vertices and guards in $G_1$ such that $g_{j_k} \in \mathbb{G}(T_{i_k}) \backslash \{ g_{j_{k-1}} \}$ and  $R_{\emptyset}^{j_k}(i_k) \subseteq R_{\emptyset}(i_k)$ such that $R_{\emptyset}^{j_k}(i_k) \subseteq \tilde{R}_{i_k}^2(j_k)$ and $R_{\emptyset}^{j_k}(i_k) \not\subseteq {R}_{i_k}^2(j_k)$. It is also stated that the sequence terminates if $\{g_{j_{k-1}}\} = \mathbb{G}(T_{i_k})$ in which case $T_{i_k}$ is an unsafe triangle. The problem of finding the existence of an allocation that works when Algorithm \ref{alg:allocation} fails is reduced to the problem of showing that the aforementioned sequence of vertices and guards in $G^{\#}$ does not terminate. Since the number of vertices in $G^{\#}$ is finite, the sequence is stuck in a cycle of vertices of $G^{\#}$. According to the definition of $G^{\#}$ there should be a cycle $\overline{C}$ in $G^{\#}$ involving the vertices and the guards of the cycle in the sequence. Now we prove that such a cycle $\overline{C}$ cannot exist unless Algorithm \ref{alg:arbitrary} was called. First, we show that for the first vertex of the cycle $\overline{C}$ that appears in the sequence, $v_{i_k}$, the pair of edges in $\overline{C}$ incident to it are both incoming edges. If $v_{i_k}$ is the vertex that corresponds to the triangle where Algorithm \ref{alg:arbitrary} determined that it could not find an allocation, the claim is trivially proved since all edges incident to $v_j$ are incoming edges. Otherwise, if $v_{i_k}$ is no such a vertex, then it corresponds to a vertex where all the incident edges are incoming edges excepting the edges that correspond to guard $g_{j_{k-1}}$. Notice that the sequence $v_{i_1}\xleftarrow{g_{j_1}}v_{i_2}\xleftarrow{g_{j_2}}\cdots$ follows a direction opposite to the orientation of the edges. Therefore, when the sequence is in $v_{i_{k}}$, the next guard cannot be $g_{j_{k-1}}$. It implies that both edges of $\overline{C}$ incident to $v_{i_{k}}$ do not correspond to $g_{j_{k-1}}$, so they are by definition incoming edges. The sequence then continues with a different guard $g_{j_{k}}$ followed by a vertex $v_{i_{k+1}}$. By definition, the edge in $\overline{C}$ incident to $v_{i_{k+1}}$ that corresponds to $g_{j_{k}}$ is an outgoing edge, so the next edge in $\overline{C}$ corresponds to a different guard $g_{j_{k+1}}$ and its corresponding edge in $\overline{C}$ is an incoming edge of $v_{i_{k+1}}$. Clearly, every vertex in $\overline{C}$ has an outgoing and an incoming edge. Since $\overline{C}$ is a cycle, vertex $v_{i_{k}}$ is eventually reached. However, it does not have an incoming and an outgoing edge in $\overline{C}$ (both are incoming edges), which is not possible. Therefore, such a cycle does not exist in $G^{\#}$. This contradicts the initial definition of $G^{\#}$. Thus, at least one edge of $\overline{C}$ was removed during the execution of Algorithm \ref{alg:allocation}. This implies that Line $8$ of Algorithm \ref{alg:arbitrary} was reached, which is impossible since Algorithm \ref{alg:arbitrary} was never called. The result follows.	
\end{proof}

\newpage

\section{List of Variables}
\label{sec:appc}
\begin{table}[H]
	\begin{tabular}{L{1.8cm}|L{6cm}}
		\midrule
		\textbf{Variable} & \textbf{Definition}\\
		\midrule
		$P$ & Polygon\\
		$n$ & Number of vertices of $P$\\
		$\bar{v}_e$ & Maximum speed of intruder\\
		$x_I$ & Location of intruder\\
		$t$ & Time\\
		$\mathbb{G}$ & Set of guards\\ 
		$\bar{v}_g$ & Maximum speed of guards\\
		$r$ & $\bar{v}_g/\bar{v}_e$\\
		$G$ &  Triangulation graph\\
		$\mathbb{V}(G)$ & Vertex set of graph $G$\\
		$\mathbb{E}(G)$ & Edge set of graph $G$\\
		$\mathbb{T}(G)$ & Faces of graph $G$\\
		$\mathbb{H}$ & Set of diagonals\\
		$g_i$ & Guard $i$\\
		$l_i$ & Length of $h_i$\\
		$v_{\alpha}(i)$ & Endpoint $\alpha$ of $h_i$\\
		$\mathbb{G}(T)$ & Set of guards incident to $T$\\
		$d(\cdot,\cdot)$ & Distance function\\
		$A(g_i)$ & Non-safe triangles allocated to $g_i$\\
		$A(\mathbb{G})$ & Allocation of all non-safe triangles\\
		$\mathcal{A}$ & Set of all $A(\mathbb{G})$\\
		$\overline{\mathbb{T}_{\alpha}^{safe}}(g_i)$ & Set of non-safe triangles incident to $v_{\alpha}(i)$\\
		$\overline{\mathbb{T}^{safe}}(G)$ & Set of non-safe triangles\\
		$G^{\#}$ & Guard adjacency graph\\
		$e_{j,k}(g_i)$ & Edge of $G^{\#}$\\
		$w_{j,k}(g_i)$ & Weight of $e_{j,k}(g_i)$\\
		$S_{rep}$ & Set of representatives of $\bigcup_{j\in\{1,\ldots, |\mathbb{V}(G^{\#})|\}} S_g(T_j)$\\	
		$R_j^{\alpha}(i)$ & Region inside $T_j$ assigned to $g_i$ incident to $v_{\alpha}(i)$\\
$S_R^{\alpha}$ & Set of regions $R_j^{\alpha}(i)$\\	
$\mathbb{R}^{alloc}(g_i)$ & Region allocated to $g_i$\\
$\hat{U}_R^{\alpha}(i)$ & Union of regions $R_j^{\alpha}(i)$\\
$c(\mathbb{R}^{alloc}(g_i))$ & Cost of $\mathbb{R}^{alloc}(g_i)$ \\
$d_I^i$ & $l_i/r$\\
$T_k^{free}$ & Region inside $T_k$ that has not been assigned\\
$R_{\emptyset}(j)$ & Region inside $T_j$ that cannot be assigned\\
$\mathbb{G}_{ready}$ & Set of guards ready to be allocated\\
$\mathbb{G}_{alloc}$ & Set of allocated guards\\
$\mathbb{G}_{\Omega}$ & Set of guards that cannot be allocated\\
$\mathbb{V}_{\Omega}$ & Set of vertices corresponding to non-safe triangles that can be covered by guards in $\mathbb{G}_{\Omega}$\\
$s_{int}^1(i)$ & Internal critical curve\\
$s_{ext}^1(i)$ & External critical curve\\
$C_1(i)$ & Critical region\\
$\hat{C}_1(i)$ & Extended critical region\\
$\mathbb{I}$ & Set of intruders
\end{tabular}
\caption{List of frequently used variables and their meaning.}
\end{table}

\end{document}